\documentclass[12pt,reqno]{amsart}

\newcommand\version{July 17, 2020}

\usepackage{amsmath,amsfonts,amsthm,amssymb,amsxtra}
\usepackage{bbm} 
\usepackage[hidelinks]{hyperref}

\newtheorem{theorem}{Theorem}
\newtheorem{proposition}[theorem]{Proposition}
\newtheorem{lemma}[theorem]{Lemma}
\newtheorem{corollary}[theorem]{Corollary}
\newtheorem{openproblem}[theorem]{Open Problem}
\newtheorem{conjecture}[theorem]{Conjecture}

\theoremstyle{definition}

\theoremstyle{remark}

\newcommand{\C}{\mathbb{C}}
\newcommand{\cl}{\mathrm{cl}}

\renewcommand{\epsilon}{\varepsilon}

\newcommand{\loc}{{\rm loc}}

\newcommand{\N}{\mathbb{N}}

\renewcommand{\phi}{\varphi}
\newcommand{\R}{\mathbb{R}}

\newcommand{\Sph}{\mathbb{S}}

\newcommand{\Z}{\mathbb{Z}}

\DeclareMathOperator{\re}{Re}

\DeclareMathOperator{\sgn}{sgn}
\DeclareMathOperator{\Tr}{Tr}

\begin{document}

\title[Lieb--Thirring inequalities --- \version]{The Lieb--Thirring inequalities:\\ Recent results and open problems}

\author{Rupert L. Frank}
\address[R. L. Frank]{Mathematics 253-37, Caltech, Pasa\-de\-na, CA 91125, USA, and Mathematisches Institut, Ludwig-Maximilans Univers\"at M\"unchen, The\-resienstr. 39, 80333 M\"unchen, Germany}
\email{rlfrank@caltech.edu}

\renewcommand{\thefootnote}{${}$} \footnotetext{\copyright\, 2020 by the author. This paper may be reproduced, in its entirety, for non-commercial purposes.\\
Partial support through US National Science Foundation grant DMS-1363432 is acknowledged.}

\maketitle

This review celebrates the generous gift by Ronald and Maxine Linde for the remodeling of the Caltech mathematics department and the author is very grateful to the editors of this volume for the invitation to contribute. We attempt to survey recent results and open problems connected to Lieb--Thirring inequalities. In view of several excellent existing reviews \cite{Li89,BlaStu-96,LaWe1,Hu,La} as well as highly recommended textbooks \cite{LiLo,LiSe}, we sometimes put our focus on developments during the past decade.

The author would like to thank all his collaborators on the topic of Lieb--Thirring inequalities and, in particular,  A.~Laptev, S.~Larson, M.~Lewin, E.~H.~Lieb, P.~T.~Nam and T.~Weidl for helpful remarks on a preliminary version of this review.

\tableofcontents


\section{The Lieb--Thirring problem}

\subsection{A Sobolev inequality for orthonormal functions}\label{sec:sobortho}

In 1975, Lieb and Thirring proved the following theorem \cite{LiTh0}, see also \cite{LiTh}.

\begin{theorem}\label{lt}
	Let $d\geq 1$. There is a constant $K_d>0$ such that for all $N\in\N$ and all functions $u_1,\ldots,u_N\in H^1(\R^d)$ that are orthonormal in $L^2(\R^d)$ one has
	\begin{equation}
	\label{eq:lt}
		\sum_{n=1}^N \int_{\R^d} |\nabla u_n|^2\,dx \geq K_d \int_{\R^d} \left( \sum_{n=1}^N |u_n|^2 \right)^{1+\frac2d}dx \,.
	\end{equation}
\end{theorem}

The main point in this theorem is that the constant $K_d$ is independent of the number $N$ of functions. Clearly, if the orthonormality requirement is dropped, then the constant on the right side would decrease like $N^{-\frac2d}$, as can be seen by taking all $u_n$ to be equal.

The work of Lieb and Thirring was motivated by giving a new proof of stability of matter and the constant $K_3$ enters into their stability estimate. We will discuss this in more detail in Section \ref{sec:som} below and the reader who wants to see the Lieb--Thirring inequality `in action', before studying its more theoretic aspects, might wish to jump directly to that section. For various other applications related to stability of matter we refer to \cite{Li89,LiSe}.

Lieb--Thirring bounds are closely connected to justifications of density functional theories, see, for instance, \cite{Li-83b,LeLiSe}. In addition, they have proved useful to bound the dimension of attractors for the Navier--Stokes flow \cite{Lieb-84}; see also \cite{CoFoTe,Te}. They also appear in the context of spectral theory of Jacobi matrices and one-dimensional Schr\"odinger operators, see, e.g., \cite{KiSi,KiSi2}.

The Lieb--Thirring theorem leads naturally to the following challenge, which is a famous open problem in the field.

\begin{openproblem}\label{op1}
	Find the optimal constant $K_d$ in \eqref{eq:lt}.
\end{openproblem}

Lieb and Thirring suggested two possible scenarios for optimality that lead to different constants and conjectured that the optimal constant $K_d$ is given by the lesser of the two constants in these scenarios. Let us describe this in more detail.


\subsection{The one-particle constant}\label{sec:onepart}

A well-known Sobolev interpolation inequality, sometimes called Gagliardo--Nirenberg inequality or Moser inequality, states that for any $d\geq 1$ there is a constant $K_d^{(1)}>0$ such that for all $u\in H^1(\R^d)$ one has
\begin{eqnarray}
\label{eq:sobint}
\int_{\R^d} |\nabla u|^2\,dx \geq K_d^{(1)} \int_{\R^d} |u|^{2+\frac4d}\,dx \left( \int_{\R^d} |u|^2\,dx \right)^{-\frac{2}{d}} \,.
\end{eqnarray}
In connection with Lieb--Thirring inequalities, the constant $K_d^{(1)}$ is called the \emph{one-particle constant}.

Clearly, choosing $N=1$ in \eqref{eq:lt}, we obtain an inequality of the form \eqref{eq:sobint} and therefore the optimal constants in these inequalities satisfy
\begin{eqnarray}
\label{eq:compop}
K_d \leq K_d^{(1)} \,.
\end{eqnarray}

Let us summarize what is known about the optimal constant $K_d^{(1)}$ and optimizers in \eqref{eq:sobint}. For detailed proofs and references we refer to \cite{Fr13}. In dimension $d=1$, one has \cite{Na}
$$
K_1^{(1)} = \frac{\pi^2}4 \,.
$$
and equality in \eqref{eq:sobint} is attained if and only if $u$ coincides, up to translation, dilation and
multiplication by a constant, with $Q(x) = \left( \cosh x \right)^{-1/2}$.

In dimensions $d\geq 2$ it is known that \eqref{eq:sobint} has an optimizer (one method of proof is suggested in \cite{LiTh} and another one is carried out in \cite{We}) and that this optimizer is radial \cite{GiNiNi} and
unique up to translation, dilation and multiplication by a constant \cite{Kw} (see also \cite{McSe} and references therein). Clearly, the optimizer $Q$ of \eqref{eq:sobint} can be normalized such that it satisfies the Euler--Lagrange equation in the form
\begin{eqnarray}
\label{eq:el}
-\Delta Q - Q^{1+4/d} = - Q
\qquad\text{in}\ \R^d \,,
\end{eqnarray}
and in this normalization, the optimal constant is related to the $L^2$ norm of $Q$ by
\begin{equation}
\label{eq:constopt}	
K_{d}^{(1)} = \frac{d}{d+2}\, \|Q\|_2^{4/d} \,.
\end{equation}
This follows by integrating \eqref{eq:el} against $Q$ and $x\cdot\nabla Q$. Moreover, $Q$ is not only the unique minimizer up to symmetries, but also the unique positive solution of \eqref{eq:el} \cite{GiNiNi,Kw}. Therefore, $Q$ can be computed numerically using the shooting method and then $K_d^{(1)}$ can be evaluated using \eqref{eq:constopt}. This computation appears in the appendix of \cite{LiTh} by Barnes for $d=2,3$.


\subsection{The semiclassical constant}\label{sec:sc}

To get a different upper bound on $K_d$, we want to choose the functions $u_n$ in \eqref{eq:lt} as plane waves. In order to make them belong to $H^1(\R^d)$ we need to multiply them by a cut-off function. Concerning their normalization, the following lemma is useful. It says that a certain relaxation of the problem does not change the optimal constant.

\begin{lemma}\label{relaxed}
	Let $(u_n)\subset H^1(\R^d)$ be a sequence of functions that are orthonormal in $L^2(\R^d)$ and let $(\nu_n)\subset[0,1]$ be a sequence of numbers. Then
	$$
	\sum_{n=1}^\infty \nu_n \int_{\R^d} |\nabla u_n|^2\,dx \geq K_d \int_{\R^d} \left( \sum_{n=1}^\infty \nu_n |u_n|^2 \right)^{1+\frac2d}dx \,, 
	$$
	where $K_d$ is the optimal constant in \eqref{eq:lt}.
\end{lemma}

\begin{proof}
	By monotone convergence, we may assume that only finitely many of the $\nu_n$'s are nonzero. We write
	$$
	\sum_{n=1}^\infty \nu_n \int_{\R^d} |\nabla u_n|^2\,dx = \int_0^1 \left( \sum_{\nu_n>\tau} \int_{\R^d} |\nabla u_n|^2\,dx\right)d\tau
	$$
	and use the bound \eqref{eq:lt} for fixed $\tau$. In this way, we obtain
	$$
	\sum_{n=1}^\infty \nu_n \int_{\R^d} |\nabla u_n|^2\,dx \geq K_d \int_{\R^d} \left( \int_0^1 \left( \sum_{\nu_n>\tau} |u_n|^2 \right)^{1+\frac2d} d\tau \right)dx \,.
	$$ 
	For fixed $x\in\R^d$, we apply H\"older's inequality in the $\tau$ integral,
	$$
	\int_0^1 \left( \sum_{\nu_n>\tau} |u_n|^2 \right)^{1+\frac2d} d\tau
	\geq \left( \int_0^1 \sum_{\nu_n >\tau} |u_n|^2 \,d\tau \right)^{1+\frac2d} = \left( \sum_n \nu_n |u_n|^2 \right)^{1+\frac2d}.
	$$
	Inserting this into the above integral leads to the claimed inequality.	
\end{proof}

Let $\chi$ be a real-valued, compactly supported, Lipschitz function satisfying $|\chi|\leq 1$ and $\chi(0)=1$. For two parameters $L,\mu>0$ we consider the integral operator $\gamma$ in $L^2(\R^d)$ with integral kernel
$$
\gamma(x,y) = \chi(x/L) \int_{|\xi|^2<\mu} e^{i\xi\cdot(x-y)}\,\frac{d\xi}{(2\pi)^d} \ \chi(y/L) \,,
\qquad x,y\in\R^d \,.
$$
We have for any $\psi\in L^2(\R^d)$
$$
\langle \psi,\gamma\psi\rangle = \int_{|\xi|^2<\mu} \left| (2\pi)^{-d/2} \int_{\R^d} e^{-i\xi\cdot y} \chi(y/L)\psi(y)\,dy \right|d\xi \,.
$$
This shows that the operator $\gamma$ is selfadjoint and nonnegative and, since the right side does not exceed
$$
\int_{\R^d} \left| (2\pi)^{-d/2} \int_{\R^d} e^{-i\xi\cdot y} \chi(y/L)\psi(y)\,dy \right|d\xi
= \int_{\R^d} \left|\chi(y/L)\psi(y)\right|^2 dy \leq \|\psi\|_2^2 \,,
$$
its norm is $\leq 1$. It is straightforward to see that $\gamma$ is compact in $L^2(\R^d)$ and that its range is contained in $H^1(\R^d)$. Therefore there are functions $u_n\in H^1(\R^d)$ that are orthonormal in $L^2(\R^d)$ and numbers $\nu_n\in(0,1]$ such that
$$
\gamma = \sum_n \nu_n |u_n\rangle\langle u_n| \,.
$$

We want to apply the inequality in Lemma \ref{relaxed} with these $u_n$ and $\nu_n$. We have
$$
\sum_n \nu_n |u_n(x)|^2 = \gamma(x,x) = \chi(x/L)^2 \int_{|\xi|^2<\mu} \,\frac{d\xi}{(2\pi)^d} = \frac{\omega_d}{(2\pi)^d} \mu^{\frac d2} \chi(x/L)^2 \,,
$$
where $\omega_d$ is the measure of the unit ball in $\R^d$. Thus,
$$
\int_{\R^d} \left( \sum_n \nu_n |u_n(x)|^2 \right)^{1+\frac2d}dx = \left( \frac{\omega_d}{(2\pi)^d} \right)^{1+\frac2d} \mu^{\frac d2 +1} L^{d} \int_{\R^d} \chi(y)^{1+\frac2d}\,dy \,.
$$
On the other hand,
\begin{align*}
\sum_n \nu_n \int_{\R^d} |\nabla u_n|^2\,dx & = \int_{|\xi|^2<\mu} \int_{\R^d} \left| \nabla \left( \chi(x/L) e^{i\xi\cdot x}\right) \right|^2dx\,\frac{d\xi}{(2\pi)^d} \\
& = \int_{|\xi|^2<\mu} \int_{\R^d} \left( |\xi|^2 \chi(x/L)^2 + L^{-2} |\nabla\chi(x/L)|^2 \right) dx\,\frac{d\xi}{(2\pi)^d} \\
& = \frac{d}{d+2} \frac{\omega_d}{(2\pi)^d} \mu^{\frac d2+1} L^{d}  \int_{\R^d} \chi(y)^2\,dy + \frac{\omega_d}{(2\pi)^d} \mu^{\frac d2}  L^{d-2} \int_{\R^d} \!|\nabla\chi(y)|^2\,dy.
\end{align*}

Thus, Lemma \ref{relaxed} with these $u_n$ and $\nu_n$ yields the inequality
\begin{align*}
\frac{d}{d+2} \int_{\R^d} \chi(y)^2\,dy + \mu^{-1}  L^{-2} \int_{\R^d} |\nabla\chi(y)|^2\,dy
\geq K_d \left( \frac{\omega_d}{(2\pi)^d} \right)^{\frac2d} \int_{\R^d} \chi(y)^{1+\frac2d}\,dy \,.
\end{align*}
In this inequality, we pass to the limit $\mu L^2 \to\infty$. This removes the gradient term and we can let $\chi$ approximate a characteristic function. In this way we finally arrive at
\begin{equation}
\label{eq:ksc}
K_d \leq \frac{d}{d+2} \frac{(2\pi)^2}{\omega_d^{\frac2d}} =: K_d^\cl \,.
\end{equation}
For an alternative proof of \eqref{eq:ksc} which does not use Lemma \ref{relaxed}, see \cite[Lemma~10]{GoLeNa}.

The constant on the right side of \eqref{eq:ksc} is called the \emph{classical (or semiclassical) constant}. The reason for this name will become clear later when discussing Weyl asymptotics for eigenvalues of Schr\"odinger operators. The discussion in this subsection is closely related to the Thomas--Fermi approximation in density functional theory, see \cite{LiSi,Li81,Li-83b,LeLiSe} and references therein.

As an aside, we mention that an analogue of the Lieb--Thirring inequality is valid for measures on phase space and is frequently used in kinetic theory; see, e.g., \cite[Eq.~(14)]{LiPe}. This analogue is valid with the classical constant.


\subsection{The Lieb--Thirring conjecture}

Lieb and Thirring \cite{LiTh} conjectured that the best constant $K_d$ in \eqref{eq:lt} is given by the worst one in the previous two scenarios.

\begin{conjecture}\label{ltconj}
	For any $d\geq 1$,
	\begin{eqnarray}
	\label{eq:ltconj}
	K_d = \min\left\{ K_d^{(1)}, K_d^{\cl} \right\}.
	\end{eqnarray}
\end{conjecture}

Using the numerical computations of $K_d^{(1)}$, one sees that this is equivalent to
$$
K_d = \begin{cases} K_d^{(1)} & \text{if}\ d=1,2\,,\\ K_d^\cl & \text{if}\ d\geq 3 \,. \end{cases}
$$


\subsection{The currently best bound}

In \cite{LiTh0}, Lieb--Thirring proved $K_3 \geq (4\pi)^{-\frac23}\, K_3^\cl$ in dimension $d=3$. Note that $(4\pi)^{-\frac23}\approx 0.185001$. Since then there have been many contributions devoted to improving the lower bounds on $K_d$ \cite{Lieb-84,EdeFoi-91,BlaStu-96,HunLapWei-00,DolLapLos-08}. In particular, Dolbeault, Loss and Laptev showed in \cite{DolLapLos-08} that $K_d \geq (4/\pi^2)^{\frac1d}\,K_d^\cl$ for all $d\geq 1$. Note that $4/\pi^2 \approx 0.405285$ and $(4/\pi^2)^\frac13\approx 0.740037$. The following bound, which is currently the best one in all dimensions, was obtained in \cite{FrHuJeNa}.

\begin{theorem}\label{ltbest}
	For all $d\geq 1$,
	\begin{eqnarray}
	\label{eq:ltbest}
	K_d \geq \left( 0.471851 \right)^\frac1d \, K_d^\cl \,.
	\end{eqnarray}
\end{theorem}

Note that in $d=3$ the excess factor is $\left( 0.471851 \right)^\frac13 \approx 0.7785137$. Moreover, in $d=1$ it is natural to compute the excess factor with respect to the conjectured constant $K_1^{(1)}=\pi^2/4$ instead of $K_1^\cl=\pi^2/3$. We have
$$
K_1 \geq \left( 0.471851 \right) \tfrac{4}{3}\, K_1^{(1)} \approx 0.629134\, K_1^{(1)} \,.
$$

The fact that the dimension enters the excess factor in the form of an exponent $1/d$ is due to the Laptev--Weidl method of lifting the dimension \cite{LapWei-00}. We will explain this technique in more detail in Subsection \ref{sec:ltbest} below. The essence is that one tries to prove an inequality similar to \eqref{eq:lt} in dimension $d=1$, but allows the functions to take values in $\C^q$ and seeks for constants independent of the `internal dimension' $q$.

An interesting bound was proved in \cite{Nam}, namely, for any $d\geq 1$ there is a $C_d<\infty$ such that inequality \eqref{eq:lt} holds with the right side replaced by
$$
(1-\epsilon) K_d^{\cl} \int_{\R^d} \left( \sum_{n=1}^N |u_n|^2 \right)^{1+\frac2d}dx - C_d\, \epsilon^{-3-\frac 4d} \int_{\R^d} \left| \nabla \left( \sum_{n=1}^N |u_n|^2 \right)^{\frac 12}\right| dx
$$
for any $\epsilon>0$. In many applications the additional gradient term is of lower order and therefore the inequality is `almost as good' as the one with the semiclassical constant. For a relevant discussion, see \cite{Li-83b}. It is an interesting open question whether the inequality holds without the factor $1-\epsilon$ in front of the first term and with a finite constant in front of the gradient term.


\section{Application: Stability of Matter}\label{sec:som}

We now recall the original motivation of the Lieb--Thirring inequality, namely the problem of stability of matter. This is the statement that the energy of a system of $N$ electrons and $K$ nuclei interacting through Coulomb forces is bounded from below by a (negative) constant times $N+K$. The constant is only allowed to depend on the charge of the nuclei.

The fact that such a bound holds was first shown by Dyson and Lenard \cite{DyLe1,DyLe2}. In \cite{LiTh0} Lieb and Thirring provided a new proof, which is shorter, gives a better bound on the involved constant and identified the Lieb--Thirring inequality as the crucial analytic ingredient in this argument. For more on the problem of stability of matter and, in particular, for further references and extensions to other physical systems, we refer to the book \cite{LiSe}.

The state of a system of $N$ quantum particles in $\R^d$ is described by a function $\psi\in L^2(\R^{dN})$, typically normalized such that $\|\psi\|_2=1$. We ignore here spin for the sake of simplicity. The corresponding one-particle density is the nonnegative function $\rho_\psi$ on $\R^d$, defined by
$$
\rho_\psi(x) = \sum_{n=1}^N \int_{\R^{d(N-1)}} |\psi(x_1,\ldots,x_{n-1},x,x_{n+1},\ldots,x_N)|^2\,dx_1\cdots dx_{n-1}\,dx_{n+1}\cdots dx_N \,.
$$
Note that
\begin{eqnarray}
\label{eq:intdensity}
\int_{\R^d} \rho_\psi(x)\,dx = N \|\psi\|_2^2 \,.
\end{eqnarray}

If the quantum particles are fermions, the function $\psi$ satisfies an additional requirement. Let $\mathcal S_N$ be the symmetric group of $\{1,\ldots,N\}$. A function $\psi:\R^{dN}\to\C$ is called \emph{antisymmetric} if for all $\sigma\in\mathcal S_N$
$$
\psi(x_{\sigma(1)},\ldots,x_{\sigma(N)}) = (\sgn\sigma)\,\psi(x_1,\ldots,x_N)
\quad\text{for all}\ x = (x_1,\ldots,x_N)\in\R^{dN} = (\R^d)^N \,.
$$
The following is a Sobolev inequality for antisymmetric functions, which is a consequence of and, in fact, equivalent to the Lieb--Thirring inequality in Theorem \ref{lt}.

\begin{corollary}\label{ltantisymm}
	Let $d\geq 1$, $N\geq 1$ and let $\psi\in H^1(\R^{dN})$ be antisymmetric. Then
	$$
	\int_{\R^{dN}} |\nabla \psi|^2\,dx_1\cdots dx_N \geq K_d \, \|\psi\|_2^{-\frac 4d} \int_{\R^d} \rho_\psi(x)^{1+\frac2d}\,dx \,,
	$$
	where $K_d$ is the optimal constant in \eqref{eq:lt}.
\end{corollary}

We emphasize that both the constant $K_d$ and the exponent $1+\frac2d$ on the right side are independent of $N$. Note that if $u_1,\ldots,u_N\in H^1(\R^d)$ are orthonormal in $L^2(\R^d)$, then for
$$
\psi(x) = \frac1{\sqrt{N!}} \det \left( u_n(x_m) \right)_{n,m=1}^N \,,
$$
the inequality in Corollary \ref{ltantisymm} reduces to that in Theorem \ref{lt}.

\begin{proof}
	We may assume that $\|\psi\|_2=1$. We define an integral operator $\gamma_\psi$ on $L^2(\R^d)$ via its integral kernel
	\begin{align*}
	\gamma_\psi(x,y) & := \sum_{n=1}^N \int_{\R^{d(N-1)}} \!\! \psi(x_1,\ldots,x_{n-1},x,x_{n+1},\ldots,x_N) \overline{\psi(x_1,\ldots,x_{n-1},x,x_{n+1},\ldots,x_N)} \\
	& \qquad\qquad\qquad dx_1\cdots dx_{n-1}\,dx_{n+1}\cdots dx_N \,.
	\end{align*}
	Thus, for any $f\in L^2(\R^d)$,
	\begin{align}\label{eq:expectationgamma}
	\langle f,\gamma_\psi f \rangle & = \sum_{n=1}^N \int_{\R^{d(N-1)}} \left|  \int_{\R^d} \overline{\psi(x_1,\ldots,x_{n-1},x,x_{n+1},\ldots,x_N)} f(x)\,dx \right|^2 \notag \\
	& \qquad\qquad\qquad dx_1\cdots dx_{n-1}\,dx_{n+1}\cdots dx_N \,.
	\end{align}
	This implies that the operator $\gamma_\psi$ is selfadjoint and nonnegative. Moreover, since
	$$
	\Tr \gamma_\psi = \sum_{n=1}^N \int_{\R^{dN}} \left| \psi(x) \right|^2\,dx = N \,,
	$$ 
	the operator is, in particular, compact and we can write
	$$
	\gamma_\psi = \sum_k \nu_k |u_k\rangle\langle u_k|
	$$
	with orthonormal functions $u_k$ and positive numbers $\nu_k$. In terms of these functions and numbers we can write
	$$
	\int_{\R^{dN}} |\nabla \psi|^2\,dx_1\cdots dx_N = \sum_k \nu_k \int_{\R^d} |\nabla u_k|^2\,dx
	$$
	and
	$$
	\rho_\psi(x) = \sum_k \nu_k |u_k(x)|^2
	\qquad \text{for a.~\!\!e.}\ x\in\R^d \,.
	$$
	Therefore, the inequality in the corollary follows immediately from the inequality in Lemma \ref{relaxed}, provided we can show that $\nu_k\leq 1$ for all $k$. This, in turn, is equivalent to the inequality
	\begin{eqnarray}
	\label{eq:densitylessthan1}
	\langle f,\gamma_\psi f \rangle \leq \|f\|_2^2
	\qquad\text{for all}\ f\in L^2(\R^d) \,,
	\end{eqnarray}
	which we shall show now. The antisymmetry of $\psi$ enters in the proof of \eqref{eq:densitylessthan1}.
	
	For the proof of \eqref{eq:densitylessthan1} we may assume that $\|f\|_2=1$. Then we may choose an orthonormal basis $(e_\alpha)_{\alpha\in\N}$ of $L^2(\R^d)$ with $e_1=f$. For $j\in\N^N$ let
	$$
	\psi^\sharp_j = \int_{\R^{dN}} \overline{e_{j_1}(x_1)\cdots e_{j_N}(x_N)} \psi(x)\,dx_1\cdots\,dx_N \,,
	$$
	so that
	$$
	\psi = \sum_{j\in\N^N} \psi_j^\sharp\ e_{j_1}\otimes\cdots\otimes e_{j_N}
	$$
	and therefore, by \eqref{eq:expectationgamma} and orthonormality,
	$$
	\langle f,\gamma_\psi f \rangle = \sum_{n=1}^N \sum_{j,j'\in\N^N} \psi^\sharp_j \overline{\psi_{j'}^\sharp} \delta_{j_1,j_1'}\cdots \delta_{j_n,1}\delta_{j_n',1}\cdots\delta_{j_N,j_N'}
	= \sum_{n=1}^N \sum_{j\in \N^N} \left| \psi^\sharp_j \right|^2 \delta_{j_n,1} \,.
	$$
	The antisymmetry of $\psi$ implies that
	$$
	\psi^\sharp_{(\sigma(j_1),\ldots,\sigma(j_N))} = (\sgn\sigma) \ \psi^\sharp_{(j_1,\ldots,j_N)}
	\qquad\text{for all}\ \sigma\in\mathcal S_N \,,\ j\in\N^N \,.
	$$
	Thus,
	$$
	\sum_{n=1}^N \sum_{j\in \N^N} \left| \psi^\sharp_j \right|^2 \delta_{j_n,1} = N! \sum_{1=j_1<j_2<\ldots<j_N} \left| \psi^\sharp_j \right|^2 \leq \sum_{j\in\N^n}\left| \psi^\sharp_j \right|^2 = \|\psi\|_2^2 = 1\,.
	$$
	This proves \eqref{eq:densitylessthan1} and completes the proof.
\end{proof}

Note that $\langle f,\gamma_\psi f \rangle = \langle\psi,K_f\psi \rangle$, where $K_f=\sum_{n=1}^N P_n$ and $P_n$ has integral kernel $f(x_n)\overline{f(x_n')}$. Thus, \eqref{eq:densitylessthan1} is equivalent to $\|K_f\|\leq \|f\|_2^2$ and, indeed, one can show that the spectrum of $K_f$, acting on antisymmetric functions in $L^2(\R^{dN})$, consists only of $0$ and $\|f\|_2^2$, which again yields \eqref{eq:densitylessthan1}.

\medspace

The energy of a Coulomb system consisting of $N$ electrons in a state described by an antisymmetric $\psi\in L^2(\R^{3N})$ and of $K$ classical nuclei at positions $R=(R_1,\ldots,R_K)\in\R^{3K}$ and with charges $Z=(Z_1,\ldots,Z_K)\in[0,\infty)^K$ is
$$
\mathcal E_{R,Z}[\psi] := \int_{\R^{3N}} \left( |\nabla \psi|^2 + V_{R,Z}(x) |\psi|^2\right)dx
$$
with the Coulomb potential
$$
V_{R,Z}(x) = -\sum_{n=1}^N \sum_{k=1}^K \frac{Z_k}{|x_n-R_k|} + \sum_{1\leq n<m\leq N} \frac{1}{|x_n-x_m|} + \sum_{1\leq k<\ell\leq K} \frac{Z_k Z_\ell}{|R_k-R_\ell|} \,.
$$
The ground state energy is
$$
E_{R,Z}(N) := \inf\left\{ \mathcal E_{R,Z}[\psi] :\ \psi\in H^1(\R^{3N}) \ \text{antisymmetric},\ \|\psi\|_2 = 1 \right\}.
$$
The following theorem states that matter is stable.

\begin{theorem}\label{som}
	For all $R_1,\ldots,R_K\in\R^3$ and $Z_1,\ldots,Z_K\in[0,\infty)$,
	$$
	E_{R,Z}(N) \geq - \frac{3\,\pi^\frac43}{2^\frac23\,5} K_3^{-1}\, (2z+1)^2\, \left( N+K\right),
	$$
	where $K_3$ is the optimal constant in \eqref{eq:lt} with $d=3$ and $z=\max_{1\leq k\leq K} Z_k$.
\end{theorem}

We outline a proof of this theorem due to Solovej \cite{So}. For another proof, which gives a slightly better value of the constant, we refer to \cite[Section 7.1]{LiSe}.

 To bound the kinetic energy contribution to $\mathcal E_{R,Z}[\psi]$ we use the Lieb--Thirring inequality in the form of Corollary \ref{ltantisymm}. To bound the potential energy contribution to $\mathcal E_{R,Z}[\psi]$ we use the following pointwise inequality due to Baxter \cite{Ba}.

\begin{lemma}\label{baxter}
	For any $x_1,\ldots,x_N,R_1,\ldots,R_K\in\R^3$ and $Z_1,\ldots,Z_K\in[0,\infty)$,
	$$
	V_{R,Z}(x) \geq - \sum_{n=1}^N \frac{2z+1}{\delta_R(x_n)} \,,
	$$
	where $\delta_R(x_n)=\min\{|x_n-R_k|:\ 1\leq k\leq K\}$ and $z=\max_{1\leq k\leq K} Z_k$.
\end{lemma}

For the proof of this lemma we refer to \cite[Theorem 5.4]{LiSe}. (There the inequality is only stated in the case where all $Z_k$'s are equal, but since for given $k$, $Z_k\mapsto V_{R,Z}(x)$ is affine linear and so attains its minimum on the interval $[0,z]$ either at $0$ or at $z$, we can reduce the general case to this special case.)

\begin{proof}[Proof of Theorem \ref{som}]
According to Corollary \ref{ltantisymm} and Lemma \ref{baxter} we have for all antisymmetric $\psi\in H^1(\R^{3N})$ with $\|\psi\|_2=1$,
$$
\mathcal E_{R,Z}[\psi] \geq K_3 \int_{\R^3} \rho_\psi^{\frac53}\,dx - \int_{\R^3} \frac{2z+1}{\delta_R} \rho_\psi\,dx \,.
$$
In the second integral on the right side we add and subtract a number $\mu>0$ from $(2z+1)\delta_R^{-1}$ and we recall \eqref{eq:intdensity}. By H\"older's inequality, we obtain
\begin{align*}
\mathcal E_{R,Z}[\psi] \geq K_3 T - \left(  \int_{\R^3} \left( \frac{2z+1}{\delta_R} - \mu\right)_+^\frac{5}{2} dx \right)^\frac25 T^\frac35 -\mu N
\qquad\text{with}\
T = \int_{\R^3} \rho_\psi^{\frac53}\,dx \,. 
\end{align*}
Optimizing the right side with respect to $T$, we obtain
$$
\mathcal E_{R,Z}[\psi] \geq - \frac{2\cdot 3^\frac32}{5^\frac52} K_3^{-\frac32} \int_{\R^3} \left( \frac{2z+1}{\delta_R} - \mu\right)_+^\frac{5}{2} dx - \mu N \,.
$$
By scaling, with $\tilde R_k = \mu (2z+1)^{-1} R_k$,
$$
\int_{\R^3} \left( \frac{2z+1}{\delta_R} - \mu\right)_+^\frac{5}{2} dx
= (2z+1)^3 \mu^{-\frac12} \int_{\R^3} \left( \frac{1}{\delta_{\tilde R}} - 1\right)_+^\frac{5}{2} dx \,.
$$
Optimizing the resulting bound in $\mu$, we obtain
$$
\mathcal E_{R,Z}[\psi] \geq - \frac{3^2}{5^\frac53} K_3^{-1} (2z+1)^2 \left( \int_{\R^3} \left( \frac{1}{\delta_{\tilde R}} - 1\right)_+^\frac{5}{2} dx \right)^\frac23 N^\frac13 \,.
$$
To complete the proof, it suffices to bound the integral in terms of $K$. Using
$$
\left( \frac{1}{\delta_{\tilde R}(x)} - 1\right)_+^\frac{5}{2}
= \max_{1\leq k\leq K} \left( \frac{1}{|x-\tilde R_k|} - 1\right)_+^\frac{5}{2}
\leq \sum_{k=1}^K \left( \frac{1}{|x-\tilde R_k|} - 1\right)_+^\frac{5}{2}
$$
and
$$
\int_{\R^3} \left( \frac{1}{|y|} - 1\right)_+^\frac{5}{2}dy = 4\pi \int_0^1 (r^{-1}-1)^\frac52 r^2\,dr = \frac{5\pi^2}{4} \,,
$$
we find
$$
\int_{\R^3} \left( \frac{1}{\delta_{\tilde R}} - 1\right)_+^\frac{5}{2} dx
\leq \sum_{k=1}^K \int_{\R^3} \left( \frac{1}{|x-\tilde R_k|} - 1\right)_+^\frac{5}{2} dx = \frac{5\pi^2}4\, K \,.
$$
Thus, we obtain
$$
\mathcal E_{R,Z}[\psi] \geq - \frac{3^2\,\pi^\frac43}{2^\frac43 \, 5} K_3^{-1} (2z+1)^2 K^\frac23 N^\frac13 \,.
$$
Since $K^\frac23 N^\frac13\leq (2^\frac23/3)(K+N)$, this yields the inequality in the theorem.
\end{proof}

Besides being of independent interest, the stability of matter theorem is used in the proof of the existence of the thermodynamic limit for Coulomb systems \cite{LeLi,LiLe}. In some applications, for instance, to thermodynamics for systems where the number of nuclei can fluctuate or even be completely random \cite{BlLe}, it is useful to have a bound as in Theorem \ref{som} but with a right side independent of $K$. This can be obtained using a variant of the above argument, see \cite[Lemma 8]{HaLeSo}.


\section{The Lieb--Thirring inequality for Schr\"odinger operators}

\subsection{Duality}

In this subsection we discuss an equivalent formulation of the Lieb--Thirring inequality in terms of eigenvalues of Schr\"odinger operators. By `equivalence' of the two inequalities we mean that the best constants in these two inequalities are in one-to-one correspondence.

We denote by $E_n(-\Delta+V)$, $n\in\N$, the negative eigenvalues of the Schr\"odinger operator $-\Delta+V$ in $L^2(\R^d)$, repeated according to multiplicities and arranged in nondecreasing order with the convention that $E_n(-\Delta+V)=0$ if $-\Delta+V$ has less than $n$ negative eigenvalues. More precisely, for real $V\in L^{1+\frac d2}(\R^d)$ one can use the Sobolev interpolation inequality \eqref{eq:sobint} to show that the quadratic form $\int_{\R^d} (|\nabla u|^2+V|u|^2)\,dx$ with form domain $H^1(\R^d)$ is lower bounded and closed and therefore induces a selfadjoint, lower bounded operator $-\Delta+V$ in $L^2(\R^d)$. The $E_n(-\Delta+V)$ are defined as the negative eigenvalues of this operator. Alternatively, and without any recourse to spectral theory, the $E_n(-\Delta+V)$ could be defined by the variational principle, which only involves the above quadratic form, but no operator \cite{LiLo}.

The Lieb--Thirring inequality for Schr\"odinger operators reads as follows, where we use the notation $a_-=\max\{-a,0\}$ for the negative part.

\begin{theorem}\label{ltpot}
	Let $d\geq 1$. There is a constant $L_d<\infty$ such that for any real $V\in L^{1+\frac d2}(\R^d)$,
	\begin{eqnarray}
	\label{eq:ltpot}
		\sum_n |E_n(-\Delta+V)| \leq L_d \int_{\R^d} V(x)_-^{1+\frac d2}\,dx \,.
	\end{eqnarray}
	Moreover, this inequality is equivalent to inequality \eqref{eq:lt} in the sense that the best constants satisfy
	\begin{equation}
	\label{eq:ltequiv}
	\left( (1+\tfrac d2) L_d \right)^{1+\frac 2d} \left( (1+\tfrac 2d) K_d \right)^{1+\frac d2} = 1 \,.
	\end{equation}
\end{theorem}

This equivalence is due to Lieb and Thirring \cite{LiTh0,LiTh}. In fact, their original proof of \eqref{eq:lt} proceeded via proving \eqref{eq:ltpot}. The first direct proof of \eqref{eq:lt} was found in \cite{EdeFoi-91} in $d=1$ and in \cite{Ru} for general $d$. Further direct proofs appeared in \cite{LuSo} (see also \cite{Nam}) and \cite{Sa}. The advantage of the reformulation \eqref{eq:ltpot} in terms of Schr\"odinger operators is that one can use what is called the Birman--Schwinger principle \cite{Bi,Schw}, which reduces the problem of estimating negative eigenvalues of an unbounded operator to the problem of estimating eigenvalues of a compact operator. This compact operator has the structure of a product of multiplication operators in position and in momentum space. As an aside, we mention that in their paper \cite{LiTh}, in order to bound operators of this form, Lieb and Thirring derived a matrix inequality that has proved very useful in many other applications as well.

\begin{proof}
	The proof is based on the variational principle for sums of eigenvalues, which says that
	\begin{eqnarray}
	\label{eq:varprinc}
		\sum_{n=1}^N \int_{\R^d} \left( |\nabla u_n|^2 + V |u_n|^2 \right)dx \geq \sum_{n=1}^N E_n(-\Delta+V)
	\end{eqnarray}
	for all $N$ and all functions $u_1,\ldots,u_N\in H^1(\R^d)$ that are orthonormal in $L^2(\R^d)$. Moreover, equality in this inequality is attained if $-\Delta+V$ has at least $N$ negative eigenvalues and the $u_n$ are the eigenfunctions corresponding to the $E_n(-\Delta+V)$ for $n=1,\ldots,N$.
	
	Thus, if \eqref{eq:ltpot} holds, it follows from \eqref{eq:varprinc} that for all $N$ and all functions $u_1,\ldots,u_N\in H^1(\R^d)$ that are orthonormal in $L^2(\R^d)$,
	$$
	\sum_{n=1}^N \int_{\R^d} |\nabla u_n|^2\,dx \geq 
	- \int_{\R^d} V \sum_{n=1}^N |u_n|^2 \,dx - L_d \int_{\R^d} V(x)_-^{1+\frac d2}\,dx \,.
	$$
	Optimizing this inequality with respect to $V\in L^{1+\frac d2}(\R^d)$ leads to the choice
	$$
	V(x) = - \left( \tfrac{d+2}2\, L_d\right)^{-\frac2d} \left( \sum_{n=1}^N |u_n|^2 \right)^\frac2d
	$$
	and to the bound
	$$
	\sum_{n=1}^N \int_{\R^d} |\nabla u_n|^2\,dx \geq 
	\tfrac{d}{d+2} \left( \tfrac{2+d}2\, L_d \right)^{-\frac2d} \int_{\R^d} \left( \sum_{n=1}^N |u_n|^2 \right)^{1+\frac2d} dx \,.
	$$
	This proves \eqref{eq:lt} with optimal constant satisfying $K_d \geq \tfrac{d}{d+2} \left( \tfrac{2+d}2\, L_d \right)^{-\frac2d}$.
	
	Conversely, if $V\in L^{1+\frac d2}(\R^d)$ is given and if $-\Delta+V$ has at least $N$ negative eigenvalues $E_n(-\Delta+V)$ with corresponding orthonormal eigenfunctions $u_n$, then by \eqref{eq:lt},
	$$
	\sum_{n=1}^N |E_n(-\Delta+V)| = - \sum_{n=1}^N \int_{\R^d} \left(|\nabla u_n|^2 + V|u_n|^2 \right)dx \leq - K_d \int_{\R^d} \rho^{1+\frac2d}\,dx + \int_{\R^d} V_- \rho\,dx 
	$$
	where $\rho = \sum_{n=1}^N |u_n|^2$. Optimizing the right side with respect to all functions $\rho$, we obtain
	$$
	\sum_{n=1}^N |E_n(-\Delta+V)| \leq \tfrac{2}{2+d} \left( \tfrac{d+2}{d} K_d \right)^{-\frac d2} \int_{\R^d} V(x)_-^{1+\frac d2}\,dx \,.
	$$
	This proves \eqref{eq:ltpot} with optimal constant satisfying $L_d\leq \tfrac{2}{2+d} \left( \tfrac{d+2}{d} K_d \right)^{-\frac d2}$. (Note that we used a similar argument already in the proof of Theorem \ref{som}.)
\end{proof}

In Subsections \ref{sec:onepart} and \ref{sec:sc} we have presented two possible scenarios for the best constant $K_d$ in \eqref{eq:lt}. Let us discuss what this means for the best constant $L_d$ in \eqref{eq:ltpot}.

\subsection{The one-particle constant}\label{sec:oneparticle}

The same argument as in the proof of Theorem \ref{ltpot} shows that there is a constant $L_d^{(1)}<\infty$ such that for all $V\in L^{1+\frac d2}(\R^d)$,
$$
|E_1(-\Delta+V)| \leq L_{d}^{(1)} \int_{\R^d} V(x)_-^{1+\frac d2}\,dx \,,
$$
and the best constant in this inequality is related to the best constant $K_d^{(1)}$ in inequality \eqref{eq:sobint} by
\begin{equation}
\label{eq:ltequiv1}
\left( (1+\tfrac d2) L_d^{(1)} \right)^{1+\frac 2d} \left( (1+\tfrac 2d) K_d^{(1)} \right)^{1+\frac d2} = 1 \,.
\end{equation}
Thus, the analogue of the constant $K_d^{(1)}$ in the dual picture is the best constant in the problem of minimizing the eigenvalue of a Schr\"odinger operator under the constraint of a fixed $L^{1+\frac d2}$-norm. This is a particular case of a question raised by Keller \cite{Ke}.

\subsection{The semiclassical constant}

We recall that for any $V\in L^{1+\frac d2}(\R^d)$, one has
\begin{equation}
\label{eq:weyl}
\lim_{\hbar\to 0} \hbar^{d} \sum_n |E_n(-\hbar^2\Delta +V)| = L_d^\cl \int_{\R^d} V(x)_-^{1+\frac d2}\,dx
\end{equation}
with
$$
L_d^\cl = \frac{2}{d+2} \, \frac{\omega_d}{(2\pi)^d} \,.
$$
Indeed, one can prove asymptotics \eqref{eq:weyl} for continuous, compactly supported $V$ using Dirichlet--Neumann bracketing and then one can use \eqref{eq:ltpot} to extend these asymptotics to all $V\in L^{1+\frac d2}(\R^d)$.

The asymptotics \eqref{eq:weyl} become more intuitive if one notices that
$$
\sum_n |E_n(-\hbar^2\Delta +V)| = \Tr\left(-\hbar^2\Delta +V\right)_-
$$
and
$$
\hbar^{-d} L_d^\cl \int_{\R^d} V(x)_-^{1+\frac d2}\,dx = \iint_{\R^d\times\R^d} \left( |\hbar\xi|^2 + V(x) \right)_- \frac{dx\,d\xi}{(2\pi)^d} \,.
$$
Thus, \eqref{eq:weyl} says that
$$
\Tr\left(-\hbar^2\Delta +V\right)_- \sim \iint_{\R^d\times\R^d} \left( |\hbar\xi|^2 + V(x) \right)_- \frac{dx\,d\xi}{(2\pi)^d}
\qquad\text{as}\ \hbar\to 0 \,.
$$
On the left side, one applies the function $E\mapsto E_-$ to the operator $-\hbar^2\Delta +V$ and then takes the trace, and on the right side, one applies the same function to the symbol $|\hbar\xi|^2 +V(x)$ and then integrates over phase space. Asymptotics \eqref{eq:weyl} say that this agrees to leading order in the semiclassical limit $\hbar\to 0$.

Note that the constants $L_d^\cl$ and $K_d^\cl$ in \eqref{eq:ksc} are related by
\begin{equation}
\label{eq:ltequivsc}
\left( (1+\tfrac d2) L_d^{\cl} \right)^{1+\frac 2d} \left( (1+\tfrac 2d) K_d^{\cl} \right)^{1+\frac d2} = 1 \,.
\end{equation}
Thus, the analogue of $K_d^\cl$ in the dual picture are Weyl asymptotics.

\subsection{The Lieb--Thirring conjecture}

Because of \eqref{eq:ltequiv}, \eqref{eq:ltequiv1} and \eqref{eq:ltequivsc}, Conjecture \ref{ltconj} can be equivalently stated as
$$
L_d = \max\left\{ L_d^{(1)},L_d^\cl\right\},
$$
that is, $L_d=L_d^{(1)}$ if $d=1,2$ and $L_d = L_d^\cl$ if $d\geq 3$.

\subsection{The currently best bound}

Because of \eqref{eq:ltequiv} and \eqref{eq:ltequivsc}, Theorem \ref{ltbest} can be equivalently stated as
\begin{eqnarray}
\label{eq:ltbestpot}
L_d \leq 1.456\ L_d^\cl
\qquad\text{for all}\ d\geq 1 \,.
\end{eqnarray}


\section{Lieb--Thirring inequalities for Schr\"odinger operators. II}

\subsection{A family of inequalities}

In \cite{LiTh} Lieb and Thirring observed that \eqref{eq:ltpot} is only a special case of a whole scale of inequalities. One has

\begin{theorem}\label{ltpotgamma}
Let $\gamma\geq \frac12$ if $d=1$, $\gamma>0$ if $d=2$ and $\gamma\geq 0$ if $d\geq 3$. There is a constant $L_{\gamma,d}<\infty$ such that for any real $V\in L^{\gamma+\frac d2}(\R^d)$,
\begin{eqnarray}
	\label{eq:ltpotgamma}
		\sum_n |E_n(-\Delta+V)|^\gamma \leq L_{\gamma,d} \int_{\R^d} V(x)_-^{\gamma+\frac d2}\,dx \,.
	\end{eqnarray}
\end{theorem}

Here, as before, $E_n(-\Delta+V)$ denote the negative eigenvalues of the Schr\"odinger operator $-\Delta+V$, which can be shown to be well-defined under the conditions on $V$ in the theorem. Moreover, we use the convention that for $\gamma=0$ the left side of \eqref{eq:ltpotgamma} denotes the number of negative eigenvalues of $-\Delta+V$, counting multiplicities.

The inequality in Theorem \ref{ltpotgamma} is due to Lieb and Thirring \cite{LiTh} in the non-endpoint cases. The case $\gamma=0$ in dimensions $d\geq 3$ is due to Cwikel \cite{Cw}, Lieb \cite{Li-76,Li-80} and Rozenblum \cite{Roz0,Roz1} and often referred to as the Cwikel--Lieb--Rozenblum (CLR) inequality. The case $\gamma=\frac12$ in dimension $d=1$ is due to Weidl \cite{Wei}. 

The work of Rozenblum on the case $\gamma=0$ completed earlier extensive work by Birman and Solomyak (summarized, for instance, in \cite{BiSo}); see also \cite{Si2} for a close-to optimal result and a conjecture. Alternative proofs of the CLR bound were obtained in \cite{Fe,LiYa,Co,Fr2,HuKuRiVu}.

According to the duality Theorem \ref{ltpot}, Open Problem \ref{op1} is equivalent to finding the best constant $L_d = L_{1,d}$ in \eqref{eq:ltpot} for $\gamma=1$. This naturally leads to

\begin{openproblem}\label{op2}
	Find the optimal constants $L_{\gamma,d}$ in \eqref{eq:ltpotgamma}.
\end{openproblem}

Motivated by the case $\gamma=1$ it is natural to look at two particular scenarios.

\subsection{The one-particle constant}\label{sec:onepartgamma}
Theorem \ref{ltpotgamma} implies, in particular, that for all $\gamma$ and $d$ as in that theorem there is a constant $L_{\gamma,d}^{(1)}<\infty$ such that for all $V\in L^{\gamma+\frac d2}(\R^d)$,
\begin{equation}
\label{eq:ltgamma1}
|E_1(-\Delta+V)|^\gamma \leq L_{\gamma,d}^{(1)} \int_{\R^d} V(x)_-^{\gamma+\frac d2}\,dx \,.
\end{equation}
Finding the optimal $L_{\gamma,d}^{(1)}$ is Keller's problem \cite{Ke} of minimizing the lowest eigenvalue of a Schr\"odinger operator $-\Delta+V$ in $\R^d$ under the constraint of a fixed $L^p$ norm of the potential $V$.

For $\gamma=0$ (which is allowed in dimensions $d\geq 3$) the interpretation of inequality \eqref{eq:ltgamma1} is that $-\Delta+V$ has no negative eigenvalue if $L_{0,d}^{(1)} \int_{\R^d} V(x)_-^{\frac d2}\,dx<1$.

For all $\gamma$ such that $\gamma+\frac d2\geq 1$, the same duality argument as in the case $\gamma=1$ shows that the optimal constant $L_{\gamma,d}^{(1)}$ is related to the optimal constant $K_{p',d}^{(1)}$ in the Sobolev interpolation inequality
\begin{equation}
\label{eq:sobintp}
\int_{\R^d} |\nabla u|^2\,dx \geq K_{p,d}^{(1)} \left( \int_{\R^d} |u|^{2p}\,dx \right)^\frac{2}{d(p-1)} \|u\|_2^{-\frac{4p}{d(p-1)}+2}
\end{equation}
by
\begin{equation}
\label{eq:ltgammaequiv1}
L_{p'-d/2,d}^{(1)} \left( K_{p,d}^{(1)} \right)^{\frac d2} = \left( \frac{d}{2p'} \right)^{\frac d2} \left( \frac{2p'-d}{2p'} \right)^{\frac{2p'-d}2}.
\end{equation}

Note that as $\gamma$ runs through the range in Theorem \ref{ltpotgamma}, the integrability exponent $2p=2(\gamma+\frac d2)/(\gamma+\frac d2-1)$ in \eqref{eq:sobintp} runs through the range $2<2p\leq\infty$ if $d=1$, $2<2p<\infty$ if $d=2$ and $2<2p\leq\frac{2d}{d-2}$ if $d\geq 3$. It is well-known that \eqref{eq:sobintp} holds precisely in this range of parameters. In particular, the nonvalidity of \eqref{eq:sobintp} for $2p=\infty$ if $d=2$ implies that \eqref{eq:ltgamma1}, and consequently \eqref{eq:ltpotgamma}, are not valid for $\gamma=0$ if $d=2$.

In dimensions $d=1$, the value of the optimal constant $K_{p,1}^{(1)}$ and the set of optimizers was determined by Nagy \cite{Na}. Via \eqref{eq:ltgammaequiv1} this leads to the formula
\begin{equation}
\label{eq:ltonepart}	
L_{\gamma,1}^{(1)} = \frac{1}{\sqrt\pi} \,\frac{\Gamma(\gamma+1)}{\Gamma(\gamma+\frac12)} \frac{(\gamma-\frac12)^{\gamma-\frac12}}{(\gamma+\frac12)^{\gamma+\frac12}}\,.
\end{equation}

In dimensions $d\geq 3$, the optimal constant in \eqref{eq:sobintp} is known for the exponent $2p=\frac{2d}{d-2}$, corresponding to $\gamma=0$, and the optimizers are known \cite{Rod,Ro,Au,Ta}.

In the remaining cases $2<2p<\frac{2d}{d-2}$ for $d\geq 2$ the optimal constants $K_{p,d}^{(1)}$ are not known explicitly. However, just like for \eqref{eq:sobint}, one can show that there is an optimizer and that this optimizer is unique up to translation, dilation and multiplication by a constant. For proofs we refer to the references mentioned in Subsection \ref{sec:onepart}. 

The duality used to derive \eqref{eq:ltgammaequiv1} also shows that optimizers $u$ of \eqref{eq:sobintp} are in one-to-one correspondence (modulo symmetries) with optimal potentials $V$ in \eqref{eq:ltgamma1}. In particular, we infer that there is a unique (up to symmetries) optimizing potential of \eqref{eq:ltgamma1}. A stability estimate for \eqref{eq:ltgamma1} was proved in \cite{CaFrLi}.

\subsection{The semiclassical constant}
Analogously to \eqref{eq:weyl} one can show that for all $\gamma$ and $d$ as in Theorem \ref{ltpotgamma} and all $V\in L^{\gamma+\frac d2}(\R^d)$,
\begin{equation}\label{eq:weylgamma}
\lim_{\hbar\to 0} \hbar^{d} \sum_n |E_n(-\hbar^2\Delta +V)|^\gamma = L_{\gamma,d}^\cl \int_{\R^d} V(x)_-^{\gamma+\frac d2}\,dx
\end{equation}
with
\begin{eqnarray}
\label{eq:constscgamma}
L_{\gamma,d}^\cl = (4\pi)^{-\frac d2} \, \frac{\Gamma(\gamma+1)}{\Gamma(\gamma+1 + \frac d2)} \,.
\end{eqnarray}
As in the special case $\gamma=1$, the relation to semiclassics becomes clearer if one writes the expression on the right side of \eqref{eq:weylgamma} as
\begin{equation}
\label{eq:phasespacegamma}
L_{\gamma,d}^\cl \int_{\R^d} V(x)_-^{\gamma+\frac d2}\,dx = \iint_{\R^d\times\R^d} \left( |\xi|^2 + V(x) \right)_-^\gamma \,\frac{dx\,d\xi}{(2\pi)^d} \,.
\end{equation}
As an aside, we mention that the validity of the Lieb--Thirring inequality is crucial for establishing asymptotics \eqref{eq:weylgamma} under the assumption $V\in L^{\gamma+\frac d2}(\R^d)$. We refer to \cite{BiLa} for examples of $V\in L^{\gamma+\frac d2}(\R^d)$ where \eqref{eq:weylgamma} fails for $\gamma=0$ in $d=2$.

\subsection{Known results}

It follows from the previous two subsections that the optimal constant $L_{\gamma,d}$ in \eqref{eq:ltpotgamma} satisfies
$$
L_{\gamma,d} \geq \max\left\{ L_{\gamma,d}^{(1)}, L_{\gamma,d}^\cl \right\}.
$$
It was originally conjectured by Lieb and Thirring \cite{LiTh} that one has equality in this inequality. This is still believed to be the case in dimension $d=1$.

\begin{conjecture}\label{conj1d}
For $d=1$ and $\gamma\geq\frac12$,
$$
L_{\gamma,1} = \max\left\{ L_{\gamma,1}^{(1)}, L_{\gamma,1}^\cl \right\}.
$$
\end{conjecture}

Using the explicit expressions for $L_{\gamma,1}^{(1)}$ and $L_{\gamma,1}^\cl$ in \eqref{eq:ltonepart} and \eqref{eq:constscgamma} one sees that Conjecture \ref{conj1d} is equivalent to
$$
L_{\gamma,1} = \begin{cases} L_{\gamma,1}^{(1)} & \text{if}\ \gamma\leq \frac32\,,\\ L_{\gamma,1}^\cl & \text{if}\ \gamma\geq \frac32 \,. \end{cases}
$$

In the following theorem we will collect all the known optimal results about the constants $L_{\gamma,d}$.

\begin{theorem}\label{ltoptimal}
The best constant $L_{\gamma,d}$ in \eqref{eq:ltpotgamma} is given by
\begin{enumerate}
\item[(a)] $L_{\gamma,d} = L_{\gamma,d}^\cl$ if $\gamma\geq \frac32$ and $d\geq 1$.
\item[(b)] $L_{\frac12,1} = L_{\frac12,1}^{(1)}$ if $\gamma=\frac12$ and $d=1$.
\end{enumerate}
\end{theorem}

Part (a) is due to Lieb and Thirring \cite{LiTh} for $\gamma=\frac{2k+1}2$, $k\in\N$, in $d=1$ and due to Aizenman and Lieb for all $\gamma\geq \frac32$ in $d=1$. The higher dimensional case in (a) is due to Laptev and Weidl \cite{LapWei-00}; see also \cite{BeLo} for an alternative proof. Part (b) is due to Hundertmark, Lieb and Thomas \cite{LiTh} with a partially alternate proof in \cite{HunLapWei-00}.

In view of Theorem \ref{ltoptimal}, Conjecture \ref{conj1d} is proved for $\gamma=\frac12$ and for $\gamma\geq \frac32$. The case $\gamma\in(\frac12,\frac32)$ is open.

Later, in Subsection \ref{sec:ltbest} we will discuss the Laptev--Weidl argument \cite{LapWei-00} used in the proof of part (a) of Theorem~\ref{ltoptimal} in $d\geq 2$.

Here, let us briefly explain the Aizenman--Lieb argument \cite{AiLi} that is used in the proof of part (a) of Theorem \ref{ltoptimal}. Clearly, for any $0\leq \gamma<\sigma$ there is a positive constant $C_{\gamma,\sigma}$ such that
$$
E_-^\sigma = C_{\gamma,\sigma} \int_0^\infty (E+\tau)_-^\gamma \tau^{\sigma-\gamma-1}\,d\tau \,.
$$
Thus,
$$
\sum_n |E_n(-\Delta+V)|^\sigma = C_{\gamma,\sigma} \int_0^\infty \sum_n |E_n(-\Delta+V+\tau)|^\gamma \tau^{\sigma-\gamma-1}\,d\tau
$$
and, in view of \eqref{eq:phasespacegamma},
$$
L_{\sigma,d}^\cl \int_{\R^d} V(x)_-^{\sigma+\frac d2}\,dx = C_{\gamma,\sigma} \int_0^\infty \left( L_{\gamma,d}^\cl \int_{\R^d} (V(x)+\tau)_-^{\gamma+\frac d2}\,dx \right) \tau^{\sigma-\gamma-1}\,d\tau \,.
$$
It follows from these two equations that the optimal constants $L_{\gamma,d}$ and $L_{\sigma,d}$ satisfy
$$
\frac{L_{\sigma,d}}{L_{\sigma,d}^\cl} \leq \frac{L_{\gamma,d}}{L_{\gamma,d}^\cl}
\qquad\text{if}\ \sigma>\gamma \,.
$$
In particular, if $L_{\gamma,d}=L_{\gamma,d}^\cl$ for some $\gamma$, then $L_{\sigma,d}=L_{\sigma,d}^\cl$ for all $\sigma>\gamma$.

\medspace

Next, let us discuss the maximum between $L_{\gamma,d}^{(1)}$ and $L_{\gamma,d}^\cl$, which appears in the original form of the Lieb--Thirring conjecture. A small variation of the Aizenman--Lieb argument together with some facts about the optimizing potential for $L_{\gamma,d}^{(1)}$ implies \cite{FrGoLe} that $\gamma\mapsto L_{\gamma,d}^{(1)}/L_{\gamma,d}^\cl$ is strictly decreasing. In Subsection \ref{sec:number} we will note that $L_{0,d}^{(1)}<L_{0,d}^\cl$ for $d\geq 8$ and therefore $L_{\gamma,d}^{(1)}<L_{\gamma,d}^\cl$ for all $\gamma\geq 0$. On the other hand, for $1\leq d\leq 7$ there is a unique $\gamma_c(d)$ where the functions $\gamma\mapsto L_{\gamma,d}^{(1)}$ and $\gamma\mapsto L_{\gamma,d}^\cl$ intersect. Indeed, for $d=1$ this is explicit and for $2\leq d\leq 7$ one can easily show that $L_{\gamma,d}^{(1)}>L_{\gamma,d}^\cl$ for small $\gamma$ (for instance, for $\gamma=1$ in $d=2$ and for $\gamma=0$ for $d\geq 3$) and that $L_{\gamma,d}^{(1)}<L_{\gamma,d}^\cl$ for large $\gamma$ (using simple trial functions in \eqref{eq:sobintp} or, alternatively, using (a) in Theorem \ref{ltoptimal}). According to the numerics from \cite{LiTh}, one has
$$
\gamma_c(d) = 
\begin{cases}
\frac32 & \text{if}\ d=1 \,,\\
1.165 & \text{if}\ d=2 \,,\\
0.8627 & \text{if}\ d=3 \,.
\end{cases}
$$

\medspace

Let us complement the `positive' results in Theorem \ref{ltoptimal} by `negative' results.

\begin{proposition}\label{ltnegative}
The best constant $L_{\gamma,d}$ in \eqref{eq:ltpotgamma} satisfies
\begin{enumerate}
\item[(a)] $L_{\gamma,d}>L_{\gamma,d}^\cl$ if $\gamma<\frac32$ in $d=1$ and $\gamma<1$ in $d\geq 2$.
\item[(b)] $L_{\gamma,d}>L_{\gamma,d}^{(1)}$ if $\gamma>\max\{2-d/2,0\}$ in $1\leq d\leq 6$ and  $\gamma\geq 0$ in $d\geq 7$.
\end{enumerate}
\end{proposition}

We have repeated the results for $d=1$ for the sake of completeness. Part (a) in $d\geq 2$ is due to Helffer and Robert \cite{HeRo}. Part (b) for $\gamma=0$ is due to Glaser, Grosse and Martin \cite{GlGrMa} and will be discussed in Subsection \ref{sec:number}. Part (b) for $\gamma>0$ is from \cite{FrGoLe}; see Subsection~\ref{sec:fgl} below for some details of the argument. 

\medspace

Let us discuss the state of the original Lieb--Thirring conjecture and some possible modifications, which take into accound the negative results from Proposition \ref{ltnegative} as well as the numerical experiments from \cite{Lev}.

\subsubsection*{Dimension $d=1$} Conjecture \ref{conj1d} is generally believed to be true. The only remaining case is the range $1/2<\gamma<3/2$, where the optimal constant should be $L_{\gamma,1}^{(1)}$.

\subsubsection*{Dimension $d=2$} In the range $0<\gamma\leq 1$, it is conceivable that $L_{\gamma,2}=L_{\gamma,2}^{(1)}$, as originally conjectured by Lieb and Thirring. (This is suggested by numerics in the appendix of \cite{LiTh} and in \cite{Lev}. However, both these computations missed the phenomena described next for $\gamma>1$, so it is not clear how reliable they are.)

The situation in the range $1<\gamma<3/2$ is rather unclear. By part (b) of Proposition~\ref{ltnegative}, one has $L_{\gamma,2}>L_{\gamma,2}^{(1)}$ for $\gamma>1$. Since $L_{\gamma,2}^{(1)}>L_{\gamma,2}^\cl$ for $\gamma<1.165$, this shows that the original Lieb--Thirring conjecture fails in this range. Moreover, it is shown in \cite{FrGoLe} that, if there is an optimizing potential for some $\gamma>1$, then this potential has infinitely many negative eigenvalues. Instead of (or besides) the existence of such an optimal potential, it is conceivable that the optimal potentials in the bound for the first $N$ eigenvalues (which exist \cite{FrGoLe}) converge, when suitably normalized, as $N\to\infty$ to a potential that does not belong to $L^{\gamma+1}$ like, for instance, a periodic potential.

\subsubsection*{Dimensions $d\geq 3$} 

Based on numerics for radial potentials, it is suggested in \cite{Lev} that for any $0\leq\gamma<1$ there is an optimal potential and that it has only a finite number of negative eigenvalues. Moreover, it is expected that the number of negative eigenvalues of an optimal potential increases as $\gamma$ increases. It is conceivable that $L_{\gamma,3}=L_{\gamma,3}^{(1)}$ for $0\leq\gamma\leq 1/2$ if $d=3$, as originally conjectured by Lieb and Thirring. On the other hand, according to Proposition \ref{ltnegative} one has $L_{\gamma,d}>\max\{ L_{\gamma,d}^{(1)},L_{\gamma,d}^\cl\}$  for $1/2<\gamma<1$ if $d=3$, for $0<\gamma<1$ if  $4\leq d\leq 6$ and $0\leq\gamma<1$ if $d\geq 7$, so in all these cases the original Lieb--Thirring conjecture fails.

According to Conjecture \ref{ltconj} and the Aizenman--Lieb argument, it is believed that $L_{\gamma,d}=L_{\gamma,d}^\cl$ for $\gamma\geq 1$ and $d\geq 3$.


\subsection{Currently best bounds}

Let us summarize bounds on the optimal constant $L_{\gamma,d}$ for $\gamma<3/2$. The best bounds in the literature are
$$
L_{\gamma,d} \leq
\begin{cases}
1.456\ L_{\gamma,d}^\cl & \text{if}\ 1\leq\gamma<3/2 \,,\\
2\ L_{\gamma,1}^\cl & \text{if}\ 1/2\leq\gamma<1 \ \text{and}\ d=1 \,,\\
2.912\ L_{\gamma,d}^\cl & \text{if}\ 1/2\leq\gamma<1 \ \text{and}\ d\geq 2\,.\\
\end{cases}
$$
By the Aizenman--Lieb argument, these bounds follow from the corresponding bounds at the smallest value at $\gamma$. Thus, the first bound follows from \eqref{eq:ltbestpot}, the second one from \cite{HuLiTh} and the third one by the Laptev--Weidl lifting argument from \cite{HunLapWei-00} and \eqref{eq:ltbestpot}. This lifting argument yields, more generally, the bound
$$
L_{1/2,d} \leq 2\ L_{1,d-1} \,.
$$

Bounds for the range $0\leq\gamma<1/2$ in $d\geq 3$ follow by the Aizenman--Lieb argument from corresponding bounds for $\gamma=0$. The best value for $L_{0,3}$ in $d=3$ is due to Lieb in \cite{Li-76,Li-80},
$$
L_{0,3} \leq 6.86924\ L_{0,3}^\cl
$$
and is to be compared with the lower bound from the Sobolev inequality $L_{0,3}\geq (8/\sqrt 3)\ L_{0,3}^\cl \approx 4.6188\ L_{0,3}^\cl$. Lieb's proof uses a new formula for Wiener integrals, called Lieb's formula, which is further discussed in \cite[Theorem~8.2]{SiFuncInt}. The best bounds for $d=4$ and for $d\geq 5$ are in \cite{Li-76,Li-80} and \cite{HuKuRiVu}, respectively. 

Bounds for the range $0<\gamma<1/2$ in $d=2$ have received relatively little attention in the literature. In particular, we are not aware of an investigation of the asymptotic behavior of $L_{\gamma,2}$ as $\gamma\to 0$. Probably, both $L_{\gamma,2}$ and $L_{\gamma,2}^{(1)}$ behave like a constant times $\gamma^{-1}$. Are the two constants the same? The asymptotics of $L_{\gamma,2}^{(1)}$ can be obtained via \eqref{eq:ltgammaequiv1} from arguments similar to those in \cite{ReWe}. A logarithmic endpoint type inequality is shown in \cite{KoVuWe}.


\subsection{The number of negative eigenvalues}\label{sec:number}

Let us discuss in more detail the (open) problem of finding the optimal constant $L_{0,d}$ for $d\geq 3$, that is, to maximize the quotient between the number of negative eigenvalues of $-\Delta+V$ and $\int_{\R^d} V_-^{\frac d2}\,dx$.

It is convenient to introduce the notation $N_\leq(-\Delta+V)$ to denote the number of nonpositive eigenvalues of $-\Delta+V$, counting multiplicities, plus the number of zero energy resonances, corresponding to solutions $u\in\dot H^1(\R^d)\setminus L^2(\R^d)$ of $(-\Delta+V)u=0$. This definition appears naturally in this context since $N_\leq(-\Delta+V)$ is the limit of the number of negative eigenvalues of $-\Delta+V_+-(1+\epsilon)V_-$ as $\epsilon\to 0+$, so inequality \eqref{eq:ltpotgamma}, even if the left side only counts negative eigenvalues, implies
$$
N_\leq(-\Delta+V)\leq L_{0,d} \int_{\R^d} V(x)_-^\frac d2\,dx \,.
$$

We begin by presenting the example of \cite{GlGrMa} that shows that $L_{0,d}>\max\{ L_{0,d}^\cl, L_{0,d}^{(1)}\}$ for $d\geq 7$. Our presentation is somewhat different from theirs and fills in some details. The basis is the following computation, which we explain later in this subsection.

\begin{lemma}\label{confex}
	Let $d\geq 3$ and, for $L\in\N_0$,
	$$
	V^{(L)}(x) = - \left( L + \frac{d-2}{2} \right)\left( L+ \frac d2\right) \left( \frac{2}{1+|x|^2} \right)^2 \,,
	\qquad x\in\R^d \,.
	$$
	Then
	$$
	N_\leq(-\Delta +V^{(L)}) = \frac{2}{d!}\ \frac{(L+d-1)!\, (L+\frac d2)}{L!}
	$$
	and
	$$
	\int_{\R^d} \left(V^{(L)}\right)_-^\frac d2\,dx = \left( (L+\tfrac{d-2}{2})(L+\tfrac d2)\right)^{\frac d2} |\Sph^d| \,.
	$$
\end{lemma}

As a consequence of this lemma,
$$
L_{0,d} \geq \sup_{L\in\N_0} \frac{N_\leq(-\Delta +V^{(L)})}{\int_{\R^d} \left(V^{(L)}\right)_-^\frac d2\,dx} = \frac{2}{d!\,|\Sph^d|}\, \sup_{L\in\N_0} a_L = L_{0,d}^\cl \, \sup_{L\in\N_0} a_L
$$
with
$$
a_L := \frac{(L+d-1)!\, (L+\frac d2)}{L!\, \left( (L+\frac{d-2}{2})(L+\frac d2)\right)^{d/2}} \,.
$$
Note that, by the form of optimizers in the Sobolev inequality \cite{Rod,Ro,Au,Ta},
$$
L_{0,d}^{(1)} = L_{0,d}^\cl \, a_0 \,.
$$
On the other hand, since $a_L\to 1$ as $L\to\infty$,
$$
L_{0,d}^\cl = L_{0,d}^\cl \, \lim_{L\to\infty} a_L \,.
$$
Thus, in order to show that $L_{0,d}>\max\{ L_{0,d}^\cl, L_{0,d}^{(1)}\}$, we need to show that $\sup_{L\in\N_0} a_L > \max\{\lim_{L\to\infty} a_L, a_0\}$. This is possible if $d\geq 7$. Indeed, as suggested to me by S.~Larson, to whom I am grateful, using $\ln(1+x)=x+\mathcal O(x^2)$ as $x\to 0$, one sees that
$$
\ln a_L = \tfrac d2 \ L^{-1} + \mathcal O(L^{-2})
\qquad\text{as}\ L\to\infty \,,
$$
so $a_L> 1 = \lim_{L'\to\infty} a_{L'}$ for all sufficiently large $L$. On the other hand, $a_1>a_0$ if $d= 7$. Since $a_0<1=\lim_{L\to\infty} a_L$ if $d\geq 8$, we have indeed shown that $\sup_{L\in\N_0} a_L > \max\{\lim_{L\to\infty} a_L, a_0\}$ for all $d\geq 7$.

Glaser, Grosse and Martin \cite{GlGrMa} make the following conjecture.

\begin{conjecture}\label{ggm}
	Let $d\geq 3$ and $\gamma=0$. Then
	$$
	L_{0,d} = L_{0,d}^\cl \, \sup_{L\in\N_0} a_L \,.
	$$
\end{conjecture}

In particular, it is conjectured that $L_{0,d}=L_{0,d}^{(1)}$ if $d\leq 6$. The CLR bound with the conjectured constant $L_{0,4}^{(1)}$ holds for radial potentials in $d=4$ \cite{GlGrMa}. Moreover, the Lieb--Thirring conjecture for $\gamma=1$ in $d=1$ would imply the CLR bound with the conjectured constant $L_{0,3}^{(1)}$ for radial potentials in $d=3$ \cite{GlGrMa}.

Further evidence for Conjecture \ref{ggm} comes from the following observation, which is analogous to one made in a related context in \cite{Fr2}, namely that the problem of computing the optimal $L_{0,d}$ is conformally invariant. More precisely, if $h$ is a conformal transformation of $\R^d\cup\{\infty\}$ with Jacobian denoted by $J_h$ and if
$$
V_h(x) = J_h(x)^{2/d}\, V(h(x)) \,,
$$
then
$$
\int_{\R^d} V_h(x)_-^\frac d2\,dx = \int_{\R^d} V(x)_-^{\frac d2}\,dx
\qquad\text{and}\qquad
N_\leq(-\Delta+V_h) = N_\leq(-\Delta+V) \,.
$$
The first equality is clear and the second one follows from the variational principle in the form (sometimes called Glazman's lemma)
\begin{align*}
& N_\leq(-\Delta+V) \\
& = \sup\left\{\dim \mathcal M:\ \mathcal M\subset\dot H^1(\R^d)\,,\ \int_{\R^d} \left(|\nabla u|^2 + V|u|^2\right)dx \leq 0 \ \forall u\in\mathcal M \right\},
\end{align*}
if we note that for $v(x) = J_h(x)^{(d-2)/(2d)} u(h(x))$ one has
$$
\int_{\R^d} |\nabla v|^2\,dx = \int_{\R^d} |\nabla u|^2\,dx \,,
\qquad
\int_{\R^d} V_h |v|^2\,dx = \int_{\R^d} V|u|^2\,dx \,.
$$
(Here, the first equality is easily verified by noting that any conformal transformation of $\R^d\cup\{\infty\}$ is a composition of a translation, a dilation, a rotation, a reflection and an inversion.)

In view of the conformal invariance it is natural to consider the optimzation problem on the sphere. We will use this procedure to prove Lemma \ref{confex}. We consider the inverse stereographic projection $\mathcal S:\R^d\to\Sph^d$,
$$
\mathcal S_j(x) = \frac{2x_j}{1+|x|^2} \,,\qquad j=1,\ldots d \,,
\qquad
\mathcal S_{d+1}(x) = \frac{1-|x|^2}{1+|x|^2} \,.
$$
Then, by a similar argument as before, if 
$$
V(x) = \left( \frac{2}{1+|x|^2} \right)^2 W(\mathcal S(x)) \,,
$$
then
$$
\int_{\R^d} V(x)_-^\frac d2\,dx = \int_{\Sph^d} W(\omega)_-^\frac d2\,d\omega
\qquad\text{and}\qquad
N_\leq(-\Delta+V) = N_\leq(-\Delta_{\Sph^d}+ \tfrac{d(d-2)}{4}+ W) \,.
$$
Here $-\Delta_{\Sph^d}$ is the Laplace--Beltrami operator on $\Sph^{d}$. Its eigenvalues are given by $\ell(\ell+d-1)$, $\ell\in\N_0$, with multiplicity 
$$
\nu_\ell = \frac{(2\ell+d-1)\ (\ell+d-2)!}{(d-1)!\ \ell!} \,.
$$
Note that the potential $V^{(L)}$ in Lemma \ref{confex} corresponds to the constant potential $W^{(L)}=- \left( L + \frac{d-2}{2} \right)\left( L+ \frac d2\right)$ on $\Sph^d$. We have
$$
\int_{\R^d} V^{(L)}(x)_-^\frac d2\,dx = \int_{\Sph^d} W^{(L)}(\omega)_-^\frac d2\,d\omega = \left( (L+\tfrac{d-2}{2})(L+\tfrac d2)\right)^{\frac d2} |\Sph^d|
$$
and, since $\ell(\ell+d-1)+\frac{d(d-2)}4 - \left( L + \frac{d-2}{2} \right)\left( L+ \frac d2\right)\leq 0$ iff $\ell\leq L$,
$$
N_\leq(-\Delta_{\Sph^d}+ \tfrac{d(d-2)}{4}+ W^{(L)}) = \sum_{\ell=0}^L \nu_\ell = \frac{2}{d!}\ \frac{(L+d-1)!\, (L+\frac d2)}{L!} \,.
$$
This completes the proof of Lemma \ref{confex}.

To summarize, Conjecture \ref{ggm} says that the optimal constant in the CLR inequality is given, after mapping the problem conformally to the sphere, by a constant potential. This would be similar to other optimization problems with conformal invariance, both for single functions \cite{Li83} and for functions of eigenvalues \cite{Mo}.



\section{Further directions of study}

After the overview over the standard Lieb--Thirring inequalities in the previous sections, we now address some extensions and generalizations. Our presentation emphasizes, probably unjustly, developments in the last decade and/or developments in which the author was involved. The overall focus is on open problems, some major, some minor, and it is hoped that the presentation stimulates further progress.


\subsection{P\'olya's conjecture}

A classical question in the field of spectral estimates concerns the best value of the constant $L_{\gamma,d}^{\rm dom}$ in the inequality
$$
\sum_n \left( E_n(-\Delta_\Omega) - \mu\right)_-^\gamma \leq L_{\gamma,d}^{\rm dom}\, |\Omega|\, \mu^{\gamma+\frac{d}{2}}
\qquad\text{for all}\ \mu\geq 0
$$
valid for all open sets $\Omega\subset\R^d$ of finite measure. Here $-\Delta_\Omega$ denotes the Dirichlet Laplacian in $\Omega$ and $E_n(-\Delta_\Omega)$ its eigenvalues in nondecreasing order, counted according to multiplicity.

Clearly, Weyl asymptotics imply that $L_{\gamma,d}^{\rm dom}\geq L_{\gamma,d}^\cl$ for all $\gamma\geq 0$. A famous conjecture by Polya states that $L_{\gamma,d}^{\rm dom}=L_{\gamma,d}^\cl$ for all $\gamma\geq 0$. (Strictly speaking, Polya only considered $\gamma=0$. By the Aizenman--Lieb argument, equality $L_{\gamma,d}^{\rm dom}=L_{\gamma,d}^\cl$ for some $\gamma=\gamma_0$ implies equality for any $\gamma>\gamma_0$. So Polya's conjecture for $\gamma=0$ implies the conjecture as stated.)

Polya has given an elegant proof of his conjectured bound in the special case of tiling domains \cite{Po}. Further results for product domains can be found in \cite{La0}.

The connection between Polya's conjecture and the Lieb--Thirring problem is that $L_{\gamma,d}^{\rm dom}\leq L_{\gamma,d}$. This follows from the variational principle by taking $V(x)=-\mu$ for $x\in\Omega$ and $V(x)\geq 0$ for $x\not\in\Omega$ in the Lieb--Thirring inequality.

Berezin \cite{Be} and Li and Yau \cite{LiYa} (the latter in an equivalent, dual form) proved that $L_{\gamma,d}^{\rm dom} = L_{\gamma,d}^\cl$ for $\gamma\geq 1$.

There has been relatively little progress on P\'olya's conjecture. In particular, it is still unknown whether the inequality holds with the semiclassical constant in the special case where $\Omega$ is a disc in $d=2$.

Some recent work concerns the analogue of P\'olya's conjecture in the presence of a homogeneous magnetic field. While the analogue of the Berezin--Li--Yau bound continues to hold in this setting \cite{ErLoVo}, the analogue of P\'olya's conjecture fails for any $0\leq\gamma<1$ \cite{FrLoWe}. Also, in \cite{KwLaSi} it was shown that the analogue of P\'olya's conjecture fails for the fractional Laplacian $(-\Delta)^s_\Omega$ in $d=1$, as well as for most $s$ in $d=2$.

Evidence for P\'olya's conjecture comes from the sign of the subleading term in Weyl's asymptotic law \cite{Ivrii_80}. Bounds that capture a lower order correction terms appear, typically for $\gamma\geq 3/2$ or $\gamma\geq 1$, in \cite{FrLiUe,Wei3,KoVuWe2,GeLaWe,FrLar} and references therein; see also \cite{Lar} for an application of these ideas to shape optimization problems.



\subsection{Magnetic Lieb--Thirring inequalities}

The Lieb--Thirring inequality in the presence of a magnetic field reads
$$
\sum_n |E_n((-i\nabla +A)^2+V)|^\gamma \leq L_{\gamma,d}^{\rm mag} \int_{\R^d} V(x)_-^{\gamma+\frac{d}{2}}\,dx  \,,
$$
where $\gamma$ is as in Theorem \ref{ltpotgamma}. By definition, $L_{\gamma,d}^{\rm mag}$ is independent of $A\in L^2_\loc(\R^d,\R^d)$.

Several of the proofs of Theorem \ref{ltpotgamma} extend to the magnetic case with the same constant. It is an open problem, however, whether the optimal constant $L_{\gamma,d}^{\rm mag}$ coincides with the optimal constant $L_{\gamma,d}$.

This is trivially the case if $d=1$, where every magnetic field can be gauged away. Moreover, Laptev and Weidl \cite{LapWei-00} showed that $L_{\gamma,d}^{\rm mag}=L_{\gamma,d}^\cl=L_{\gamma,d}$ for $\gamma\geq 3/2$ in any dimension $d$. All bounds that are obtained using their method starting from a one-dimensional inequality remain valid in the magnetic case, including the current best bound \eqref{eq:ltbestpot}. There is a semi-abstract result \cite{Fr1}, which says that $L_{\gamma,d}^{\rm mag}$ does not exceed $L_{\gamma,d}$ by more than a factor depending only on $\gamma$ and $d$. This result is also applicable to spectral inequalities of a more complicated form than Lieb--Thirring inequalities.

The case of the Pauli operator, that is, $(\sigma\cdot(-i\nabla+A))^2$ instead of $(-i\nabla+A)^2$, is considerably more difficult and we refer to \cite{Er,LiLoSo,Sob1,BuFeFrGrSt,ErSo1,ErSo2,ErSo3} and references therein.


\subsection{Lieb--Thirring inequalities for powers of the Laplacian}

The Lieb--Thirring inequality for powers $s>0$ of the Laplacian reads
$$
\sum_n |E_n((-\Delta)^s+V)|^\gamma \leq L_{\gamma,d,s} \int_{\R^d} V(x)_-^{\gamma+\frac{d}{2s}}\,dx \,,
$$
where
$$
\begin{cases}
\gamma\geq 1-\frac{d}{2s} & \text{if}\ d < 2s \,,\\
\gamma> 1-\frac{d}{2s} & \text{if}\ d=2s \,,\\
\gamma\geq 0 & \text{if}\ d>2s \,.
\end{cases}
$$
The inequality in the cases $\gamma>(1-d/2s)_+$ and $\gamma=0$ can be proved using the methods from \cite{LiTh} and \cite{Cw,Roz1}, respectively. The inequality for $\gamma=1-d/2s>0$ appears in \cite{Wei,NeWe} for integer $s$ and in \cite{Fr18} for $s< 1$. The proof for noninteger $s>1$ should follow along the same lines.

While some of the above proofs yield reasonably good constants $L_{\gamma,d,s}$, nothing seems to be known about their optimal values for $s\neq 1$. In particular, one might wonder whether $L_{\gamma,d,s}$ coincides with its semiclassical analogue for sufficiently large $\gamma$.

On the other hand, for $d=1$ and any integer $s\geq 2$ it is shown in \cite{FoOs} that in the critical case $\gamma=1-d/2s$, the optimal constant $L_{\gamma,d,s}$ is strictly larger than the corresponding one-particle constant, contrary to a conjecture in \cite{LaWe1}. One might wonder whether it is equal to the one-particle constant in $d=1$ for $s\in(1/2,3/2)$.


\subsection{Lieb--Thirring inequalities for discrete Schr\"odinger operators}

Results for Jacobi matrices and discrete Schr\"odinger operators can be found, for instance, in \cite{HuSi,KiSi,RoSo,Sah,Schi,BadSLa} and in the references therein. Due to the lack of scaling invariance the form of the inequality and therefore also the question of optimal constants is less clear in this setting.


\subsection{The oval problem}

The Lieb--Thirring Conjecture \ref{conj1d} would imply, in particular, that $|E_1(-\tfrac{d^2}{dx^2}+V)|+|E_2(-\tfrac{d^2}{dx^2}+V)|$ is bounded by $L_{1,1}^{(1)} \int_{\R} V_-^\frac 32\,dx$. Benguria and Loss \cite{BeLo2} reformulated this weaker conjecture as an isoperimetric problem for certain planar curves and proved an initial result. Further progress is contained in \cite{BuTh,Li,BeMe,De}, but the problem is still open.


\subsection{Semiclassical monotonicity}

Remarkably, in \cite{Stu} it was shown that the function $\hbar\mapsto\hbar^{-d} \sum_n |E_n(-\hbar^2\Delta+V)|^\gamma$ is nonincreasing for $\gamma\geq 2$ and $d\geq 1$. Moreover, taking $V(x)=|x|^2-1$ and $\hbar$ near $(d+2)^{-2}$, one sees that the assumption $\gamma\geq 2$ is necessary.


\subsection{Reverse Lieb--Thirring inequalities}

In \cite{GlGrMa,Sc} the Lieb--Thirring bound for $\gamma=1/2$ in $d=1$ is complemented by the lower bound
$$
\sum_n |E_n(-\tfrac{d^2}{dx^2}+V)|^{1/2} \geq -L_{1/2,1}^\cl \int_\R V(x)\,dx
$$
with optimal constant $L_{1/2,1}^\cl=1/4$. Similar bounds for $V\leq 0$ were proved for $0<\gamma< 1/2$ if $d=1$ \cite{DaRe} and for $\gamma=0$ if $d=2$ \cite{Ko,GrNeYa}; see also \cite{NeWe,Sh1}. While for most of these bounds, optimal (or almost optimal) values of the constants have not been investigated, remarkably, for $\gamma=0$ in $d=2$ one has the optimal inequality
$$
N_\leq(V) \geq 1 + \left\lfloor \left( \frac{1}{8\pi} \int_{\R^3} V(x)\,dx \right)_- \,\right\rfloor
$$
For $V\leq 0$ this follows by conformal invariance as in Subsection \ref{sec:number} from the corresponding result on $\Sph^2$ in \cite{KaNaPePo}, which also contains references to earlier partial results. As in \cite{GrNaSi} the bound extends to not necessarily nonpositive $V$. In particular,
$$
N_\leq(V) \geq \left( \frac{1}{8\pi} \int_{\R^3} V(x)\,dx \right)_-,
$$
which, in the radial case, goes back to \cite{GlGrMa}, 

A completely unrelated form of a reverse Lieb--Thirring inequality is shown in \cite{DoFeLoPa}, namely, the inequality in Theorem \ref{ltpotgamma} for $\gamma<-d/2$. The constant is the classical one. This follows by integrating the Golden--Thompson inequality \cite{Go,Sy,Th}.


\subsection{Bounds on the number of negative eigenvalues in 2D}

The CLR inequality does not hold for $\gamma=0$ in $d=2$ and there have been many attempts of finding suitable analogues. Phenomena one has to deal with are the existence of weakly coupled bound states \cite{Si} as well the existence of $L^1$ potentials with non-Weyl asymptotics \cite{BiLa}. 

Contributions to this area include \cite{Sol,Wei2,KhMaWu,St,MoVa,KoVuWe,Sh,GrNa,LaSo,LaSo2,FrLa}. In particular, the paper \cite{LaSo} raises the question of characterizing all $V\in L^1(\R^2)$ (or all $0\geq V\in L^1(\R^2)$) such that either $\limsup_{\alpha\to\infty} \alpha^{-1} N(-\Delta+\alpha V)<\infty$ or such that \eqref{eq:weylgamma} with $d=2$ and $\gamma=0$ holds. This problem was solved in the radial case in \cite{LaSo}, but is still open in general. The eigenvalue bounds in \cite{Sol,KhMaWu,Sh,GrNa,LaSo2} can be understood as sufficient conditions for an asymptotically linear bound.


\subsection{Hardy--Lieb--Thirring inequalities}

These are bounds where the operator $-\Delta$ is replaced by an operator $-\Delta -w$ with a function (Hardy weight) $w\geq 0$ such that $-\Delta-w\geq 0$. For the case $w(x)=(d-2)^2/(4|x|^2)$, as well as its extensions to powers of the Laplacian and magnetic fields, we refer to \cite{EkFr,FrLiSe2,Fr0} and, for applications to the problem of stability of relativistic matter in magnetic fields, to \cite{FrLiSe}. For bounds on domains where $w$ blows up at the boundary, see \cite{FrLo,GeLaWe}, and for the fractional Pauli operator, see \cite{BlFo}.


\subsection{Equivalence of Sobolev and Lieb--Thirring inequalities}

While it is clear that Lieb--Thirring inequalities imply Sobolev (interpolation) inequalities, it is quite remarkable that, in an abstract setting under certain assumptions, the converse implication holds as well. This was shown in \cite{LeSo} for CLR inequalities and extended in \cite{FrLiSe2,FrLiSe3} to LT inequalities. The analogue of Weidl's result for $\gamma=1/2$ \cite{Wei} is missing in this abstract framework.



\subsection{Lieb--Thirring inequalities at positive density}

In \cite{FrLeLiSe0,FrLeLiSe} Lieb--Thirring inequalities were extended to the case of a positive, constant background density or, equivalently, to the case of potentials that tend to a positive constant at infinity. Informally, the inequalities can be written as
\begin{align*}
& \Tr\left( (-\Delta+V-\mu)_-^\gamma - (-\Delta-\mu)_-^\gamma + \gamma (-\Delta-\mu)_-^{\gamma-1} V \right) \\
& \quad \leq L_{\gamma,d}' \int_{\R^d} \left( (V-\mu)_-^{\gamma+\frac d2} - \mu^{\gamma+\frac d2} + \left(\gamma+\tfrac d2\right) \mu^{\gamma+\frac d2-1} V \right)dx
\end{align*}
with $\mu>0$. With a suitable interpretation of the left side, these inequalities were shown in \cite{FrLeLiSe} for $\gamma\geq 1$ in dimensions $d\geq 2$. Conditions under which the difference $(-\Delta+V-\mu)_-^\gamma - (-\Delta-\mu)_-^\gamma$ is trace class where given in \cite{FrPu}, see also \cite{FrPu1,FrPu2}.

These Lieb--Thirring inequalities have found applications in the study of quantum many body systems at positive density, for instance, in \cite{LeSa2,LeSa1,LeSa3}.

The optimal values of the constants $L_{\gamma,d}'$ are not known. Are they semiclassical for $\gamma\geq 3/2$?

Moreover, for $\gamma=1$ in $d=1$ the inequality holds only with a logarithmic correction term. Does the inequality hold without this term for $\gamma>1$?

Are there similar inequalities for $\gamma<1$? In particular, for $\gamma=0$ this is related to bounds for the spectral shift functions; see, e.g., \cite{Sob,Pu}.


\subsection{Lieb--Thirring inequalities for interacting systems}

The paper \cite{LuSo} initiated the study of Lieb--Thirring inequalities for interacting systems. These inequalities generalize the form of the Lieb--Thirring inequality in Corollary \ref{ltantisymm}, but the left side now takes into account interactions between the particles and the antisymmetry requirement is modified. Some works on this topic are \cite{FrSe,LuSo2,LuPoSo,LuNaPo,LuSe,LaLuNa} and the references therein.


\subsection{Lieb--Thirring inequalities for complex potentials}

Recently, there has been some interest in extending Lieb--Thirring inequalities to the case of complex-valued potentials. It is known (see, e.g. \cite{Fr3}) that if $V\in L^{\gamma+d/2}(\R^d)$ with $\gamma$ as in Theorem \ref{ltpotgamma}, then $-\Delta+V$ can be defined as an $m$-sectorial operator and its spectrum in $\C\setminus[0,\infty)$ consists of isolated eigenvalues of finite algebraic multiplicity.

Keller-type inequalities, that is, bounds on eigenvalues in terms of the $L^p$ norm of the potential appeared first in \cite{AbAsDa}. Bounds on sums of eigenvalues outside a cone around the positive axis were proved in \cite{FrLaLiSe}. The Laptev--Safronov conjecture \cite{LaSa} concerns the optimal range of Keller-type inequalities and is still open; see \cite{Fr11,FrSi} for some results. For Keller-type bounds with other norms than $L^p$ norms, see, for instance, \cite{DaNa,Saf,FaKrVe,LeSe,Cu} and references therein.

In \cite{Bo} it is shown that for any $\gamma>d/2$ there is a bounded $V\in L^{\gamma+d/2}(\R^d)$ such that $-\Delta+V$ has infinitely many eigenvalues in the lower halfplane that accumulate at every point in $[0,\infty)$. Whether such $V$ exist even for $\gamma>1/2$ in $d\geq 2$ is open.

Bounds on sums of powers of eigenvalues are typically obtained either by identifying eigenvalues with zeros of an analytic function and then using tools from complex analysis, or by operator theoretic techniques. We refer to \cite{DeHaKa,BoGoKu,Ha,DeHaKa2,FrSa1,FrLaSa,Fr3} and references therein. Despite these works, the natural form of the Lieb--Thirring inequality in the complex case seems to be unclear; see \cite{BoSt} for a counterexample in $d=1$ to one possible form.

Much earlier, Pavlov \cite{Pa1,Pa2} has shown that the threshold between finitely and infinitely many eigenvalues, which is a $|x|^{-2}$ decay in the real case, becomes a $\exp(-c\sqrt{|x|})$ decay in the case of a complex potential. For a bound on the number of negative eigenvalues in the analogous problem for Jacobi matrices, see \cite{BoFrVo}.


\subsection{Inequalities for orthonormal systems}

Lieb \cite{Li-83a} showed that if $0<\alpha<d/2$ and if $f_1,\ldots,f_N$ are orthonormal in $L^2(\R^d)$, then
$$
\left\| \sum_{n=1}^N \left|(-\Delta)^{-\frac\alpha 2} f_n\right|^2 \right\|_{\frac{d}{d-2\alpha}} \leq C_{d,\alpha}\, N^\frac{d-2\alpha}{d} \,.
$$
This is a strengthening of the Hardy--Littlewood--Sobolev inequality, which concerns the case $N=1$. The important feature of this bound is that $N$ appears on the right side with an (optimal) exponent $\frac{d-2\alpha}{d}<1$. Without orthogonality, the exponent would be $1$. For an alternative proof, see \cite{Ru0}, and for a conjecture about the optimal constant, see \cite{Fr2}. For related inequalities, see \cite{HoKwYo,GoLeNa}.

In \cite{FrLeLiSe1,FrSa1,FrSa2} a similar extension of the Strichartz inequality to orthonormal functions was proved. For instance, if $p,q\geq 1$ satisfy $\frac 2p+\frac dq=d$ and $1\leq q<1+\frac 2{d-1}$ and if $f_1,\ldots,f_N$ are orthonormal in $L^2(\R^d)$, then
$$
\left\| \sum_{n} \nu_n \left|e^{it\Delta} f_n\right|^2 \right\|_{L^p_t L^q_x} \leq C_{d,q} \|\nu\|_{\ell^{\frac{2q}{q+1}}} \,.
$$
For further results and open problems, see \cite{BeHoLeNaSa,Nak,BeLeNa1,BeLeNa2}. For applications of these bounds to the dynamics of quantum many-body systems, see, for instance, \cite{LeSa2,LeSa1}.

The papers \cite{FrSa1,FrSa3,Fr3} also contain further bounds on orthonormal systems related to Fourier restriction estimates. These have applications to bounds on eigenvalues of Schr\"odinger operators with complex potentials.


\section{Some proofs}


\subsection{Proof of Theorem \ref{ltbest}}\label{sec:ltbest}

The following result is due to \cite{FrHuJeNa}.

\begin{theorem}\label{ruminimprovedmatrix}
	Let $d\geq 1$ and $q\geq 1$. Let $N\in\N$ and let $u_1,\ldots,u_N\in H^1(\R^d,\C^q)$ be orthonormal in $L^2(\R^d,\C^q)$. Then
	$$
	\sum_{n=1}^N \int_{\R^d} |\nabla u_n(x)|_{\C^q}^2\,dx \geq \tilde K_d  \int_{\R^d} \Tr_{\C^q} \left( R(x)^{1+2/d} \right) dx \,,
	$$
	where $R(x)$ is the Hermitian nonnegative $q\times q$ matrix given by
	$$
	R(x) = \sum_{n=1}^N u_n(x) u_n(x)^*
	$$ 
	and where
	$$
	\tilde K_d =  \frac{2^{6/d}\, d^2\, (2\pi)^2}{(d+2)^{2+4/d}\, |\Sph^{d-1}|^{2/d} }\ \mathcal I_d^{-2/d} = \frac{2^{6/d}\, d^{1-2/d}}{(d+2)^{1+4/d}}\ \mathcal I_d^{-2/d}\ K_d^\cl
	$$
	with
	\begin{align*}
	\mathcal I_d = & \inf\left\{ \left( \int_0^\infty w(s)^2\,ds \right)^{d/2} \int_0^\infty \frac{(1-g(t))^2}{ t^{1+d/2}}\,dt :\ f,\,w \geq 0 \,,\
	\int_0^\infty f(s)^2\,ds=1, \right. \\
	& \qquad\qquad\qquad\qquad\qquad\qquad\qquad\qquad\qquad\left. g(t)= \int_0^\infty w(s)f(st)\,ds \right\}.
	\end{align*}
\end{theorem}

\begin{proof}
	\emph{Step 1.} Let $f$ be a nonnegative function on $(0,\infty)$ with $\int_0^\infty f(s)^2\,ds =1$. For any $E>0$ we define functions $u_1^E,\ldots, u_N^E$ by
	\begin{equation*}
	\widehat{u_n^E}(\xi) = f(E/|\xi|^2)\widehat{u_n} (\xi)
	\qquad\text{for all}\ \xi\in\R^d \,.
	\end{equation*}
	Then, since
	$$
	|\xi|^2 = \int_0^\infty f(E/|\xi|^2)^2\,dE
	\qquad\text{for all}\ \xi\in\R^d \,,
	$$
	we have
	\begin{align}
	\label{eq:ruminproof0}
	\sum_{n=1}^N \int_{\R^d} |\nabla u_n(x)|_{\C^q}^2\,dx
	& = \sum_{n=1}^N \int_{\R^d} |\xi|^2 |\widehat{u_n}(\xi)|_{\C^q}^2\,d\xi
	= \sum_{n=1}^N \int_{\R^d} \int_0^\infty |\widehat{u_n^E}(\xi)|_{\C^q}^2 \,dE\,d\xi \notag \\
	& = \int_{\R^d} \sum_{n=1}^N \int_0^\infty |u_n^E(x)|_{\C^q}^2 \,dE\,dx \,.
	\end{align}
	Our goal will be to bound $\sum_{n=1}^N \int_0^\infty |u_n^E(x)|_{\C^q}^2 \,dE$ from below pointwise in $x$. We note that in Rumin's original argument \cite{Ru}, $\sum_{n=1}^N |u_n^E(x)|_{\C^q}^2$ is bounded pointwise in $x$ and $E$. The additional integration in $E$, however, allows us to improve the constant.
	
	\emph{Step 2.} 
	Let $R^E(x)$ be the Hermitian nonnegative $q\times q$ matrix given by
	$$
	R^E(x) = \sum_{n=1}^N u_n^E(x) u_n^E(x)^* \,.
	$$ 
	Moreover, let $w$ be a nonnegative, square-integrable function on $(0,\infty)$ and let
	$$
	g(t):= \int_0^\infty w(s) f(st)\,ds \,.
	$$
	The crucial step in the proof will be to show that for any $x\in\R^d$, $\epsilon>0$ and $\mu>0$ one has, in the sense of matrices,
	\begin{align}
	\label{eq:ruminproofmain}
	R(x) \leq (1+\epsilon)\mu^{-1} \|w\|^2 \int_0^\infty R^E(x)\,dE + (1+\epsilon^{-1}) A \mu^{d/2} \,,
	\end{align}
	where
	\begin{equation*}
	A = 2^{-1}(2\pi)^{-d} |\Sph^{d-1}| \int_0^\infty (1-g(t))^2 t^{-1-d/2}\,dt \,.
	\end{equation*}
	
	In order to prove \eqref{eq:ruminproofmain}, for any $E>0$ let $v_1^E,\ldots v_N^E\in L^2(\R^d,\C^q)$ be defined by
	$$
	\widehat{v_n^E}(\xi) = g(E/|\xi|^2) \widehat{u_n}(\xi)
	\qquad\text{for all}\ \xi\in\R^d \,.
	$$	
	For $e\in\C^q$, $n\in\{1,\ldots,N\}$ and $\mu>0$ we bound
	\begin{align}
	\label{eq:ruminproof1}
	\left| e^* u_n(x) \right|^2 & = \left| e^* v_n^\mu(x) \right|^2 + 2\re \overline{e^* v_n^\mu(x)} \ e^*\left(u_n(x)-v_n^\mu(x)\right)  + \left| e^* (u_n(x) - v_n^\mu(x)) \right|^2 \notag \\
	& \leq (1+\epsilon) \left| e^* v_n^\mu(x) \right|^2 + (1+\epsilon^{-1}) \left| e^* (u_n(x) - v_n^\mu(x)) \right|^2.
	\end{align}
	For the first term on the right side we have, by the Schwarz inequality,
	\begin{align*}
	|e^*v_n^\mu(x)|^2 & = (2\pi)^{-d} \left| \int_{\R^d} \int_0^\infty e^{i\xi\cdot x} w(s) f(s\mu/|\xi|^2) e^* \widehat{u_n}(\xi)\,ds\,d\xi \right|^2 \\
	& \leq \|w\|^2 (2\pi)^{-d} \int_0^\infty \left| \int_{\R^d} e^{i\xi\cdot x} f(s\mu/|\xi|^2) e^*\widehat{u_n}(\xi)\,d\xi \right|^2 ds \\
	& = \|w\|^2 \int_0^\infty |e^* u_n^{s\mu}(x)|^2\,ds \\
	& = \mu^{-1} \|w\|^2 \int_0^\infty |e^* u_n^{E}(x)|^2\,dE \,.
	\end{align*}
	Inserting this into \eqref{eq:ruminproof1} and summing over $n$, we obtain
	\begin{align*}
	e^* R(x) e & = \sum_n \left| e^* u_n(x) \right|^2 \\
	& \leq (1+\epsilon) \mu^{-1} \|w\|^2 \int_0^\infty \sum_{n=1}^N |e^* u_n^{E}(x)|^2\,dE
	+ (1+\epsilon^{-1}) \sum_{n=1}^N \left| e^* (u_n(x) - v_n^\mu(x)) \right|^2 \\
	& = (1+\epsilon) \mu^{-1} \|w\|^2 \int_0^\infty e^* R^E(x) e\,dE
	+ (1+\epsilon^{-1}) \sum_{n=1}^N \left| e^* (u_n(x) - v_n^\mu(x)) \right|^2 \,.
	\end{align*}
	To bound the second term on the right side, we write
	$$
	e^* (u_n(x) - v_n^\mu(x)) = (2\pi)^{-d/2} \int_{\R^d} e^{i\xi\cdot x} (1-g(\mu/|\xi|^2)) e^* \widehat{u_n}(\xi)\,d\xi = \left( \chi_{x,\mu} e, \widehat{u_n} \right),
	$$
	where the last inner product is in $L^2(\R^d,\C^q)$ and where
	$$
	\chi_{x,\mu}(\xi):= (2\pi)^{-d/2} e^{-i\xi\cdot x} (1-g(\mu/|\xi|^2))
	\qquad\text{for all}\ \xi\in\R^d \,.
	$$
	Since the $\widehat{u_n}$ are orthonormal, we obtain by Bessel's inequality
	$$
	\sum_{n=1}^N \left| e^* (u_n(x) - v_n^\mu(x))\right|^2 \leq \| \chi_{x,\mu} e \|^2 = \tilde A \mu^{d/2} |e|_{\C^q}^2 \,,
	$$
	where
	\begin{equation*}
	\tilde A := (2\pi)^{-d} \int_{\R^d} (1-g(1/|\eta|^2))^2 \,d\eta = A \,.
	\end{equation*}
	To summarize, we have shown that
	$$
	e^* R(x) e \leq (1+\epsilon) \mu^{-1} \|w\|^2 \int_0^\infty e^* R^E(x) e\,dE
	+ (1+\epsilon^{-1}) A \mu^{d/2} |e|_{\C^q}^2 \,,
	$$
	which is the same as \eqref{eq:ruminproofmain}.
	
	\emph{Step 3.} We denote by $\lambda_j(H)$, $j=1,\ldots,q$, the eigenvalues of a Hermitian $q\times q$ matrix $H$, arranged in nonincreasing order and counted according to multiplicities. Then, by the variational principle, the matrix inequality \eqref{eq:ruminproofmain} implies that
	$$
	\lambda_j(R(x)) \leq (1+\epsilon)\mu^{-1} \|w\|^2 \lambda_j\left(\int_0^\infty R^E(x)\,dE\right) + (1+\epsilon^{-1}) A \mu^{d/2} 
	\qquad\text{for}\ j=1,\ldots,q \,.
	$$
	Optimizing with respect to $\epsilon>0$ and $\mu>0$ for each $j$, we obtain
	$$
	\lambda_j(R(x)) \leq \left(\frac{2}{d}\right)^\frac{2d}{d+2} \left( 1+ \frac d2 \right)^2 \|w\|^\frac{2d}{d+2} A^\frac{2}{d+2} \left( \lambda_j\left(\int_0^\infty R^E(x)\,dE\right) \right)^\frac{d}{d+2} \,,
	$$
	which is the same as
	$$
	\lambda_j\left(\int_0^\infty R^E(x)\,dE\right) \geq \left( \frac{d}{2} \right)^2 \left( 1+ \frac d2 \right)^{-\frac{2(d+2)}{d}} \|w\|^{-2} A^{-\frac{2}{d}} \left( \lambda_j(R(x)) \right)^{1+\frac2d} \,.
	$$
	Thus,
	\begin{align*}
	\int_0^\infty \sum_{n=1}^N |u_n^E(x)|_{\C^q}^2\,dE & = \Tr_{\C^q} \int_0^\infty R^E(x)\,dE = \sum_{j=1}^q \lambda_j\left(\int_0^\infty R^E(x)\,dE\right) \\
	& \geq \left( \frac{d}{2} \right)^2 \left( 1+ \frac d2 \right)^{-\frac{2(d+2)}{d}} \|w\|^{-2} A^{-\frac{2}{d}} \sum_{j=1}^q \left( \lambda_j(R(x)) \right)^{1+\frac2d} \\
	& = \left( \frac{d}{2} \right)^2 \left( 1+ \frac d2 \right)^{-\frac{2(d+2)}{d}} \|w\|^{-2} A^{-\frac{2}{d}} \Tr_{\C^q} \left( R(x)^{1+\frac2d} \right) .
	\end{align*}
	Inserting this bound into \eqref{eq:ruminproof0} we obtain the claimed bound.
\end{proof}

We now prove an upper bound on $\mathcal I_1$ by choosing appropriate trial functions $f$ and $w$. We also prove a lower bound on $\mathcal I_1$, which shows the limitation of the method.

\begin{lemma}\label{ruminnumerics}
	If $d=1$, then
	$$
	\frac23 \leq \mathcal I_1 \leq 0.747112 \,.
	$$
	In particular,
	\begin{equation}
	\label{eq:ruminnumerics}
		\tilde K_1 \geq (1.456)^{-2}\, K_1^\cl \,.
	\end{equation}
\end{lemma}

\begin{proof}
	Let $f,w\geq 0$ with $\int_0^\infty f(s)^2\,ds =1$ and denote $a:=\int_0^\infty w(s)^2\,ds$. Then
	$$
	g(t) := \int_0^\infty w(s) f(st)\,ds \leq \left( \int_0^\infty w(s)^2\,ds \right)^{1/2} \left( \int_0^\infty f(st)^2\,ds \right)^{1/2} = a^{1/2} t^{-1/2}
	$$
	and therefore $|1-g(t)|\geq (1-a^{1/2} t^{-1/2})_+$ for all $t>0$. Thus,
	$$
	\int_0^\infty \frac{(1-g(t))^2}{t^{3/2}}\,dt \geq \int_0^\infty \frac{(1-a^{1/2} t^{-1/2})_+^2}{t^{3/2}}\,dt = \frac{2}{3} a^{-1/2} \,,
	$$
	which implies that $\mathcal I_1 \geq \frac23$, as claimed.
	
	In order to prove the upper bound on $\mathcal I_1$ we choose
	$$
	f(s) = (1+\mu_0 s^{4.5})^{-0.25} \,,
	\qquad
	w(s) = c_0 \frac{(1- s^{0.36})^{2.1}}{1+s} \chi_{[0,1]}(s) \,,
	$$
	where $\mu_0$ and $c_0$ are determined by $\int_0^\infty f(s)^2\,ds=\int_0^\infty w(s)\,ds = 1$. A numerical computation leads to the claimed bound on $\mathcal I_1$.
\end{proof}

We now complete the proof of Theorem \ref{ltbest}. In fact, we will argue by duality and prove \eqref{eq:ltbestpot}. This will follow by dualizing Theorem \ref{ruminimprovedmatrix} and applying the Laptev--Weidl method of lifting the dimension \cite{LapWei-00}.

Let $\mathcal G$ be a separable Hilbert space. We first observe that the inequality in Theorem \ref{ruminimprovedmatrix} remains valid, with the same constant $\tilde K_d$, for functions $u_1,\ldots,u_N\in H^1(\R^d,\mathcal G)$ that are orthonormal in $L^2(\R^d,\mathcal G)$. This follows by a simple approximation argument using a sequence of finite dimensional projections on $\mathcal G$ that converges strongly to the identity.

Next, by the same duality argument as in the proof of Theorem \ref{ltpot} we infer that
\begin{eqnarray}
\label{eq:ltlw}
\sum_n |E_n(-\Delta+W)| \leq \tilde L_d \int_{\R^d} \Tr_\mathcal G \left( W(x)_-^{1+\frac{d}{2}}\right)dx
\end{eqnarray}
for any measurable function $W$ from $\R^d$ into the selfadjoint operators on $\mathcal G$ such that $W(x)_-^{1+\frac{d}{2}}$ is trace class for almost every $x$ and such that the integral of its trace is finite. The operator $-\Delta+W$ on the left side of \eqref{eq:ltlw} acts in $L^2(\R^d,\mathcal G)$. The constant $\tilde L_d$ is related to the constant $\tilde K_d$ in Theorem \ref{ruminimprovedmatrix} by
\begin{equation}
\label{eq:ltequivlw}
\left( (1+\tfrac d2) \tilde L_d \right)^{1+\frac 2d} \left( (1+\tfrac 2d) \tilde K_d \right)^{1+\frac d2} = 1 \,.
\end{equation}

Now let $V\in L^{1+\frac d2}(\R^d)$ be real. We introduce coordinates $x=(x',x_d)\in\R^{d-1}\times\R$ in $\R^d$ and write
$$
-\Delta + V = -\tfrac{d^2}{dx_d^2} + W
\qquad\text{in}\ L^2(\R^d) = L^2(\R,\mathcal G)
$$
where $W$ acts as `multiplication' by
$$
W(x_d) := - \Delta' + V(\cdot,x_d)
\qquad\text{in}\ \mathcal G := L^2(\R^{d-1}) \,. 
$$
Applying inequality \eqref{eq:ltlw} with $d=1$, we obtain
$$
\sum_n |E_n(-\Delta+V)| = \sum_n |E_n(-\tfrac{d^2}{d^2x_d} + W)| \leq \tilde L_1 \int_{\R} \Tr_{\mathcal G}\left( W(x_d)_-^{\frac{3}{2}}\right)dx_d \,.
$$
Moreover, by the Lieb--Thirring inequality of Laptev and Weidl \cite{LapWei-00}, for any $x_d\in\R$,
$$
\Tr_{\mathcal G}\left( W(x_d)_-^{\frac{3}{2}}\right)
= \sum_n |E_n(-\Delta' + V(\cdot,x_d))|^{3/2} \leq L_{3/2,d-1}^\cl \int_{\R^{d-1}} V(x',x_d)_-^{\frac32+\frac{d-1}2}\,dx' \,.
$$
Inserting this into the above bound we obtain
\begin{eqnarray}
\label{eq:fhjnlw}
\sum_n |E_n(-\Delta+V)| \leq \tilde L_1\, L_{3/2,d-1}^\cl \int_{\R^{d}} V(x)_-^{\frac32+\frac{d-1}2}\,dx \,.
\end{eqnarray}

It remains to bound the constant on the right side using the explicit bound on $\tilde K_1$ from Lemma \ref{ruminnumerics}. We first note that
$$
L_{1,1}^\cl L_{3/2,d-1}^\cl = L_d^\cl \,.
$$
This follows either using the explicit expression for $L_{\gamma,d}^\cl$ together with identities for gamma functions or, more conceptually, from
\begin{align*}
L_{d}^\cl & = \int_{\R^d} (|\xi|^2-1)_- \,\frac{d\xi}{(2\pi)^d}
= \int_{\R^{d-1}} \int_\R \left( (\xi')^2 + \xi_d^2 - 1 \right)_- \,\frac{d\xi_d}{2\pi} \, \frac{d\xi'}{(2\pi)^{d-1}} \\
& = \int_{\R^{d-1}} \left( (\xi')^2 -1 \right)_-^{\frac32} \int_\R (\eta_d^2 -1)_- \frac{d\eta_d}{2\pi} \, \frac{d\xi'}{(2\pi)^{d-1}} 
= L_{3/2,d-1}^\cl L_{1,1}^\cl \,,
\end{align*}
where we changed variables $\xi_d = ((\xi')^2-1)_-^{1/2}\eta_d$. Moreover, by \eqref{eq:ltequivsc} and \eqref{eq:ltequivlw},
$$
\left( \tilde L_1 / L_1^\cl \right)^{3} \left( \tilde K_1 / K_1^\cl \right)^{\frac 32} = 1 \,.
$$
Thus, by \eqref{eq:ruminnumerics},
$$
\tilde L_1\, L_{3/2,d-1}^\cl = \frac{\tilde L_1}{L_1^\cl}\ L_d^\cl = \left( \frac{K_1^\cl}{\tilde K_1} \right)^{\frac 12} L_d^\cl \leq 1.456 \, L_d^\cl \,.
$$
Inserting this into \eqref{eq:fhjnlw} we obtain \eqref{eq:ltbestpot}, as claimed.


\subsection{Proof of (b) in Proposition \ref{ltnegative}}\label{sec:fgl}

Here we focus on the assertions in part (b) of Proposition \ref{ltnegative} concerning $\gamma>0$. Those concerning $\gamma=0$ have already been discussed in Subsection \ref{sec:number}.

We follow \cite{FrGoLe}. Let us define $L_{\gamma,d}^{(N)}$ to be the best constant in the inequality
\begin{eqnarray}
\label{eq:ltpotgamman}
\sum_{n=1}^N |E_n(-\Delta+V)|^\gamma \leq L_{\gamma,d}^{(N)} \int_{\R^d} V(x)_-^{\gamma+\frac d2}\,dx \,.
\end{eqnarray}
Then, clearly, $L_{\gamma,d}^{(N)}\leq L_{\gamma,d}^{(N+1)}$ and $L_{\gamma,d}= \lim_{N\to\infty} L_{\gamma,d}^{(N)}$. The remaining part of assertion (b) in Proposition \ref{ltnegative} is therefore a consequence of the following result.

\begin{proposition}\label{binding2}
	If $\gamma>\max\{2-d/2,0\}$, then $L^{(2)}_{\gamma,d}>L^{(1)}_{\gamma,d}$.
\end{proposition}

Define $p$ by $p'=\gamma+d/2$ and recall from Subsection \ref{sec:onepartgamma} that inequality \eqref{eq:sobintp} has an optimizer $Q$. After a translation, a dilation and multiplication by a constant we can assume that $Q$ is positive, centered at the origin and satisfies
$$
-\Delta Q - Q^{2p-1} = -Q 
\qquad\text{in}\ \R^d \,.
$$
We abbreviate $m:=\int_{\R^d} Q^2\,dx$ and record two identities for the function $Q$, namely,
\begin{align}
\label{eq:pohozaev1}
& \int_{\R^d} |\nabla Q|^2\,dx - \int_{\R^d} Q^{2p}\,dx = -m
\qquad\qquad\text{and} \\
\label{eq:pohozaev2}
& \left( \frac d2-1 \right)\int_{\R^d} |\nabla Q|^2\,dx - \frac d{2p} \int_{\R^d} Q^{2p}\,dx = - \frac d2 m \,.
\end{align}
They follow by multiplying the equation for $Q$ by $Q$ and by $x\cdot\nabla Q$, respectively.

Next, for a parameter $R>0$, let
$$
Q_\pm(x) = Q(x\pm(\tfrac R2,0))
$$
and
$$
V = - (Q_+^2 + Q_-^2)^{p-1} \,.
$$
The main ingredient of the proof of Proposition \ref{binding2} is the following bound.

\begin{lemma}
	For any $\epsilon>0$, as $R\to\infty$,
	$$
	|E_1(-\Delta+V)|^\gamma + |E_2(-\Delta+V)|^\gamma \geq 2 \left( 1+ \frac{\gamma}{m} A + o(A) \right)
	$$
	where
	$$
	A = \frac12 \int_{\R^d} \left( (Q_+^2+Q_-^2)^p - Q_+^{2p} - Q_-^{2p} \right)dx \to 0 \,.	
	$$
\end{lemma}

Before proving this lemma, let us use it to deduce Proposition \ref{binding2}. We clearly have
$$
\int_{\R^d} V_-^{\gamma+d/2}\,dx = 2 \int_{\R^d} Q^{2p}\,dx - 2A \,,
$$
so
\begin{align*}
\frac{|E_1(-\Delta+V)|^\gamma + |E_2(-\Delta+V)|^\gamma}{\int_{\R^d} V_-^{\gamma+d/2}\,dx} & \geq \frac{1}{\int_{\R^d} Q^{2p}\,dx}\  \frac{1+ \frac{\gamma}{m} A + o(A)}{1- \left( \int_{\R^d} Q^{2p}\,dx \right)^{-1} A} \\
& = L_{\gamma,d}^{(1)} \left( 1 + \left( \gamma - \frac{m}{\int_{\R^d} Q^{2p}\,dx} \right) \frac{A}{m} + o(A) \right).
\end{align*}
Here we used the fact that 
$$
\frac{1}{\int_{\R^d} Q^{2p}\,dx} = L_{\gamma,d}^{(1)} \,,
$$
which follows from \eqref{eq:ltgammaequiv1}, \eqref{eq:pohozaev1} and \eqref{eq:pohozaev2}. Using the latter two identities again, we find
$$
\gamma - \frac{m}{\int_{\R^d} Q^{2p}\,dx} = \frac{\gamma}{p} \,,
 $$
and therefore
$$
L_{\gamma,d}^{(2)} \geq L_{\gamma,d}^{(1)} \left( 1 + \frac\gamma p \frac{A}{m} + o(A) \right),
$$
which completes the proof of the proposition. Thus, it remains to prove the lemma.

\begin{proof}
	Clearly, $E:= \int_{\R^d} Q_+ Q_-\,dx \to 0$ as $R\to\infty$, and therefore in the following we may assume that $|E|<m$. Then the two functions $\psi^{(\pm)}$ defined by
	$$
	\begin{pmatrix}
	\psi^{(+)} \\ \psi^{(-)}
	\end{pmatrix}
	:= \begin{pmatrix}
	m & E \\ E & m
	\end{pmatrix}^{-1/2} 
	\begin{pmatrix}
	Q_+ \\ Q_- 
	\end{pmatrix}
	$$
	are orthonormal in $L^2(\R^d)$.	Let
	$$
	\mathcal H := \begin{pmatrix}
	\langle \psi^{(+)},(-\Delta+V)\psi^{(+)}\rangle & \langle \psi^{(+)},(-\Delta+V)\psi^{(-)}\rangle \\
	\langle \psi^{(-)},(-\Delta+V)\psi^{(+)}\rangle & \langle \psi^{(-)},(-\Delta+V)\psi^{(-)}\rangle
	\end{pmatrix}.
	$$
	By the variational principle, the two lowest eigenvalues of $-\Delta+V$ are not larger than the corresponding eigenvalues of $\mathcal H$ and therefore, in particular,
	$$
	|E_1(-\Delta+V)|^\gamma + |E_2(-\Delta+V)|^\gamma \geq \Tr \mathcal H_-^\gamma \,.
	$$
	We have
	$$
	\mathcal H = h + \begin{pmatrix}
	0 & \delta \\ \delta & 0
	\end{pmatrix}
	$$
	with
	$$
	h = \langle \psi^{(+)},(-\Delta+V)\psi^{(+)}\rangle = \langle \psi^{(-)},(-\Delta+V)\psi^{(-)}\rangle
	$$
	and
	$$
	\delta = \langle \psi^{(+)},(-\Delta+V)\psi^{(-)}\rangle = \langle \psi^{(-)},(-\Delta+V)\psi^{(+)}\rangle \,.
	$$
	It is easy to see that $h\to -1$ and $\delta\to 0$ as $R\to\infty$, and therefore
	$$
	\Tr \mathcal H_-^\gamma = 2|h|^{\gamma} - \gamma |h|^{\gamma-1} \Tr  \begin{pmatrix}
	0 & \delta \\ \delta & 0
	\end{pmatrix} 
	+ \mathcal O(\delta^2) 
	= 2|h|^{\gamma} + \mathcal O(\delta^2) \,.
	$$
	
	Next, let us expand $h$. We have
	\begin{align*}
	|\nabla\psi^{(+)}|^2 + |\nabla\psi^{(-)}|^2 & = \frac{m}{m^2-E^2} \left( |\nabla Q_+|^2 + |\nabla Q_-|^2 \right)  - \frac{2E}{M^2-E^2} \nabla Q_+ \cdot \nabla Q_-
	\end{align*}
	and therefore, using the equation for $Q$,
	\begin{align*}
	\int_{\R^d} \left( |\nabla\psi^{(+)}|^2 + |\nabla\psi^{(-)}|^2 \right)dx & = -2 + \frac{2m}{m^2-E^2} \int_{\R^d} Q^{2p}\,dx \\
	& \ \quad - \frac{E}{m^2-E^2} \int_{\R^d} \left( Q_+^{2p-2}+Q_-^{2p-2}\right)Q_+ Q_- \,dx \,.
	\end{align*}
	Similarly,
	$$
	(\psi^{(+)})^2 + (\psi^{(-)})^2 = \frac{m}{m^2-E^2} \left( Q_+^2 + Q_-^2 \right)  - \frac{2E}{M^2-E^2} Q_+ Q_-
	$$
	and therefore
	\begin{align*}
	h & = \frac12 \left( \langle \psi^{(+)},(-\Delta+V)\psi^{(+)}\rangle + \langle \psi^{(-)},(-\Delta+V)\psi^{(-)}\rangle \right) \\
	& = -1 - \frac{m}{m^2-E^2} A + \frac{E}{m^2-E^2} B \,,
	\end{align*}
	where $A$ is as in the lemma and
	$$
	B := \int_{\R^d} Q_+Q_- \left( (Q_+^2+Q_-^2)^{p-1} - \frac12 \left( Q_+^{2p-2} + Q_-^{2p-2}\right)\right)dx \,.
	$$
	Using $Q(x)\leq C (1+|x|)^{-(d-1)/2} e^{-|x|}$ we can bound
	$$
	E = \mathcal O_\epsilon(e^{-(1-\epsilon)R})
	\qquad\text{and}\qquad
	B = \mathcal O_\epsilon(e^{-(1-\epsilon)R})
	$$
	and obtain
	$$
	|h|^\gamma = (-h)^\gamma = (1+ m^{-1}A)^\gamma + \mathcal O_\epsilon(e^{-(2-\epsilon)R}) = 1+ \gamma m^{-1} A + \mathcal O(A^2) + \mathcal O_\epsilon(e^{-(2-\epsilon)R}) \,. 
	$$
	This gives the desired expansion of $h$ expansion.
	
	Next, we show $\delta = \mathcal O_\epsilon(e^{-(2-\epsilon)R})$. By a similar computation as before, we find that
	$$
	\nabla\psi^{(+)}\cdot\nabla\psi^{(-)} = - \frac{E}{m^2 -E^2} \frac 12 \left( |\nabla Q_+|^2 + | \nabla Q_-|^2 \right) + \frac{m}{m^2-E^2} \nabla Q_+\cdot\nabla Q_-
	$$
	and therefore, by the equation,	
	\begin{align*}
	\int_{\R^d}\! \!\nabla\psi^{(+)}\!\cdot\!\nabla\psi^{(-)} dx = - \frac{E}{m^2-E^2} \!\int_{\R^d}\! Q^{2p} \,dx  + \frac{m}{m^2-E^2} \frac12 \!\int_{\R^d} \!Q_+Q_- (Q_+^{2p-2}\!+\!Q_-^{2p-2})\,dx.
	\end{align*}
	Moreover,
	$$
	\psi^{(+)}\psi^{(-)} = - \frac{E}{m^2 -E^2} \frac 12 \left(Q_+^2+ Q_-^2\right) + \frac{m}{m^2-E^2} Q_+ Q_-
	$$
	and therefore
	$$
	\int_{\R^d} V\psi^{(+)}\psi^{(-)}\,dx = \frac{E}{m^2-E^2} \frac12 \!\int_{\R^d} (Q_+^2+Q_-^2)^p\,dx - \frac{m}{m^2-E^2} \!\int_{\R^d} Q_+ Q_- (Q_+^2+Q_-^2)^{p-1}\,dx.
	$$
	Thus,	
	\begin{align*}
	\delta & = \frac12 \left( \langle \psi^{(+)},(-\Delta+V)\psi^{(-)}\rangle + \langle \psi^{(-)},(-\Delta+V)\psi^{(+)}\rangle \right) = \frac{E}{m^2 -E^2} A - \frac{m}{m^2 -E^2} B \,. 
	\end{align*}
	Together with the above bounds on $E$ and $B$, this gives the claimed bound on $\delta$.
	
	To summarize, so far we have shown that
	$$
	|E_1(-\Delta+V)|^\gamma + |E_2(-\Delta+V)|^\gamma \geq 2 \left( 1+ \frac{\gamma}{m} A + \mathcal O_\epsilon (e^{-(2-\epsilon)R}) + \mathcal O(A^2) \right)
	$$
	To get a lower bound on $A$ we use $Q(x)\geq c (1+|x|)^{-(d-1)/2} e^{-|x|}$. Therefore \cite{GoLeNa}, the integrand in the definition of $A$ is $\geq c_\epsilon e^{-(p+\epsilon)R}$ if $|x|\leq 1$, which gives
	$$
	A \geq c_\epsilon' e^{-(p+\epsilon)R} \,.
	$$
	This dominates the error term $\mathcal O_\epsilon(e^{-(2-\epsilon)R})$ if $p<2$ (that is, $\gamma+d/2>2$) and completes the proof of the lemma.	
\end{proof}

A stronger conclusion than that in Proposition \ref{binding2} can be shown under the additional assumption $\gamma\geq 1$. Namely, in \cite{FrGoLe} it is shown that for such $\gamma$ there is a sequence $N_j\to\infty$ such that $L_{\gamma,d}^{(N_j)}<L_{\gamma,d}^{(N_{j+1})}$ for all $j$. Therefore, in particular,
$$
L_{\gamma,d}^{(N)} < L_{\gamma,d}
\qquad\text{for all}\ N\geq 1
\ \text{if}\
\begin{cases}
	\gamma>3/2 & \text{if}\ d=1 \,,\\
	\gamma>1 & \text{if}\ d=2 \,,\\
	\gamma\geq 1 & \text{if}\ d\geq 3 \,.
\end{cases}
$$
This shows that the best Lieb--Thirring constant cannot be attained for a potential having finitely many negative eigenvalues.

The proof of this stronger conclusion uses the following equivalence, which generalizes \eqref{eq:ltgammaequiv1} to general $N$, provided $\gamma\geq 1$.

\begin{lemma}\label{dualityn}
	Let $1\leq\gamma<\infty$ and $1<p\leq 1+\frac 2d$ be related by $\gamma = p'-\frac d2$, and let $N\in\N$. Then inequality \eqref{eq:ltpotgamman} is equivalent to the inequality
	\begin{equation*}
	\sum_{n=1}^N \|\nabla u_n\|^2 \geq K_{p,d}^{(N)} \left( \int_{\R^d} \left( \sum_{n=1}^N |u_n|^2 \right)^p dx \right)^{\frac 2{d(p-1)}} \left( \sum_{n=1}^N \|u_n\|^{\frac{2(1-(1-2/d)p)}{1+2/d-p}} \right)^{-\frac 2{d(p-1)}+1}
	\end{equation*}
	for all $(u_n)\subset H^1(\R^d)$ that are orthogonal in $L^2(\R^d)$, in the sense that the optimal constants satisfy
	$$
	L_{p'-d/2,d}^{(N)} \left( K_{p,d}^{(N)} \right)^{\frac d2} = \left( \frac{d}{2p'} \right)^{\frac d2} \left( \frac{2p'-d}{2p'} \right)^{\frac{2p'-d}2}.
	$$
\end{lemma}

Note that the $u_n$ are orthogonal and not necessarily orthonormal. Lemma \ref{dualityn} follows by an argument similarly as in the proof of Theorem \ref{ltpot}, see \cite{FrGoLe}. The analogue corresponding to $N=\infty$ can be found in \cite{LiPa}.


\end{document}